\documentclass[journal, twocolumn, twoside]{IEEEtran}
\usepackage{algorithmic}
\usepackage{algorithm}
\usepackage{amsmath}
\usepackage{amsfonts}
\usepackage{amssymb}
\usepackage{array}
\usepackage{booktabs}
\usepackage{cite}
\usepackage{CJK}
\usepackage{color}
\usepackage{comment}
\usepackage{epsfig}
\usepackage{epstopdf}
\usepackage{graphics}
\usepackage{graphicx}
\usepackage{indentfirst}
\usepackage{latexsym}
\usepackage{listings}
\usepackage{pifont}
\usepackage[caption = false]{subfig}
\usepackage{textcomp}
\usepackage{url}
\usepackage{verbatim}
\usepackage{xcolor}

\newtheorem{theorem}{Theorem}
\newtheorem{assumption}{Assumption}

\newtheorem{corollary}[theorem]{Corollary}

\newtheorem{definition}[theorem]{Definition}

\newtheorem{lemma}[theorem]{Lemma}

\newtheorem{remark}[theorem]{Remark}

\newenvironment{proof}[1][Proof]{\emph{#1. }}{\hfill\ensuremath{\blacksquare}}

\begin{document}

\title{Compressive Toeplitz Covariance Estimation From Few-Bit Quantized Measurements With Applications to DOA Estimation}

\author{Hongwei Xu, Weichao Zheng, and Zai Yang%
\thanks{\emph{Corresponding author: Zai Yang}.}
\thanks{The authors are with the School of Mathematics and Statistics, Xi'an Jiaotong University, Xi'an 710049, China (e-mails: xhwice@stu.xjtu.edu.cn, zwcabc@stu.xjtu.edu.cn, yangzai@xjtu.edu.cn).}
\thanks{This work was supported by the National Natural Science Foundation of China under grants 12526209 and 12371464.}
}

\maketitle

\begin{abstract}
    This paper addresses the problem of estimating the Hermitian Toeplitz covariance matrix under practical hardware constraints of sparse observations and coarse quantization.
    Within the triangular-dithered quantization framework, we propose an estimator called Toeplitz-projected sample covariance matrix (Q-TSCM) to compensate for the quantization-induced bias, together with its finite-bit counterpart termed the $2k$-bit Toeplitz-projected sample covariance matrix ($2k$-TSCM), obtained by truncating the pre-quantization observations.
    Under the complex Gaussian assumption, we derive non-asymptotic error bounds of the estimators that reveal a quadratic dependence on the quantization level and capture the effect of sparse sampling patterns through the so-called coverage coefficient.
    To further improve performance, we propose the quantized sparse and parametric approach (Q-SPA) based on a covariance-fitting criterion, which enforces additionally positive semidefiniteness at the cost of solving a semidefinite program.
    Numerical experiments are presented that corroborate our theoretical findings and demonstrate the effectiveness of the proposed estimators in the application to direction-of-arrival estimation.
\end{abstract}

\section{Introduction}
\IEEEPARstart{T}{he} covariance matrix characterizes second-order dependence among variables, and its estimation is a fundamental task in multivariate statistics \cite{eldar2020sample, wright2022high, johnstone2001distribution} and signal processing \cite{krim2002two, yang2019source, dirksen2025subspace}, underpinning a wide range of downstream applications.
Furthermore, in many practical settings, such as direction-of-arrival (DOA) estimation \cite{yang2018sparse} and radar image processing \cite{aubry2021structured}, the data are commonly modeled as (generalized) stationary over a uniform grid, meaning that the covariance between two measurements depends only on their separation.
This motivates the assumption that the covariance matrix $\boldsymbol{T} \in \mathbb{C}^{d\times d}$ possesses a Hermitian Toeplitz structure, which can theoretically enhance estimation performance by reducing the number of free parameters from $\mathcal{O}(d^2)$ to $\mathcal{O}(d)$ \cite{cai2013optimal}.

In this context, we encounter practical constraints.
On the one hand, in applications such as DOA estimation, fully sampled spatial observations are typically associated with uniform linear arrays (ULAs).
Deploying large-scale ULAs requires many array elements, incurring not only high hardware costs but also increased susceptibility to mutual coupling effects \cite{patole2017automotive, liao2012doa}.
To alleviate these issues, sparse linear arrays (SLAs) are often employed in practice \cite{li2025sparse}, implying that \textit{only a sparse subset of array positions is observed}.
On the other hand, practical systems must account for the effects of analogue-to-digital conversion (ADC), i.e., the quantization of continuous-valued observations.
High-resolution ADCs are typically costly and power-hungry \cite{walden2002analog}.
Consequently, in large-scale or low-power systems, subsequent processing and inference tasks are often carried out on \textit{coarsely quantized data} \cite{hu2025model, studer2016quantized}.

Therefore, this paper focuses on Toeplitz covariance estimation under sparse observations and coarse quantization.
Specifically, let $\mathcal{D}$ be a distribution on $\mathbb{C}^d$ with zero mean and Hermitian Toeplitz covariance matrix $\boldsymbol{T} \in \mathbb{C}^{d\times d}$, and let $\boldsymbol{z}^{(l)} \in \mathbb{C}^d$ denote independent samples drawn from $\mathcal{D}$.
We denote by a ruler $\Omega$ an index set whose pairwise differences cover all distance offsets in $[d]$ (corresponding to SLA sampling).
Assume that only the sub-vectors $\boldsymbol{z}^{(l)}_\Omega$ indexed by the ruler $\Omega$ are available, and that, after ADC, we observe only their quantized versions $\dot{\boldsymbol{z}}^{(l)}_\Omega$.
The problem addressed in this paper is to estimate the covariance matrix $\boldsymbol{T}$ from these incomplete quantized observations $\dot{\boldsymbol{z}}^{(l)}_\Omega, l = 1, 2, \ldots, n$.

Research on Toeplitz covariance estimation from infinite-precision measurements is well established.
Many approaches, such as ruler-based methods \cite{eldar2020sample}, sparse-recovery schemes based on the Vandermonde decomposition of $\boldsymbol{T}$ \cite{gonen2016subspace}, and techniques built upon sparse Fourier transforms \cite{eldar2020sample}, have been developed.
Furthermore, distribution-free methods based on covariance-fitting criteria, such as semiparametric iterative covariance-based estimation (SPICE) \cite{stoica2010spice} and the sparse and parametric approach (SPA) \cite{yang2014discretization}, have been proposed and have gained popularity.
Closely related to these efforts, compressive covariance sensing (CCS) estimates second-order statistics directly from incomplete measurements by exploiting structural properties such as Toeplitzness \cite{romberg2009compressive, romero2015compressive}.
It shows that accurate covariance recovery can be achieved without signal reconstruction, provided the sampling operator satisfies appropriate coverage conditions.
In radar and array signal processing, maximum likelihood estimation (MLE), expectation-maximization (EM), and related algorithms have also been extensively studied and applied \cite{yang2023robust, du2020toeplitz, aubry2024advanced, aubry2021structured}.
However, the aforementioned methods do not explicitly model the additional bias induced by coarse quantization, which makes them nontrivial to apply to our setting.

In recent years, research on quantized covariance estimation has made considerable progress.
Building on the classical arcsine law, \cite{dirksen2022covariance, maly2022new} provide an in-depth analysis of covariance estimation from one-bit quantized data and establish corresponding non-asymptotic error bounds.
It is shown in \cite{dirksen2022covariance} that the analysis also applies to masked covariance estimation when the sampling pattern is known, whereas \cite{chen2023high} extends the results to the case of unknown sampling patterns.
Moreover, \cite{lu20251} proposes a 1.5-bit quantization framework, offering additional design flexibility for low-bit quantization.
For multi-bit covariance estimation, \cite{chen2023quantizing} introduces a triangular-dithering quantization framework and provides a rigorous non-asymptotic analysis of the resulting estimation error.
Regarding the choice of quantization (or dithering) levels, both \cite{chen2025parameter} and \cite{dirksen2024tuning} conduct in-depth investigations and propose data-driven strategies that substantially alleviate the need for manual tuning.
However, the above line of work primarily characterizes quantization effects for the sample covariance matrix (SCM) and does not explicitly incorporate additional structural priors, which limits its applicability to sparse observation settings.

Another related approach to mitigating the performance loss caused by coarse quantization is oversampling. 
Prior studies have shown that oversampling can improve information-theoretic performance in quantized communication systems, such as information rate and capacity per unit cost \cite{gilbert1993increased, koch2010increased}. 
In one-bit massive MIMO systems, oversampling has also been shown to improve detection, achievable-rate, and channel-estimation performance \cite{uccuncu2018oversampling, shao2020channel, shao2021dynamic}. 
These results suggest that the additional temporal observations provided by oversampling may reduce the number of bits required to achieve a target estimation accuracy in certain settings. 
Incorporating this idea into our setting would require accounting for the temporal dependence associated with oversampled measurements, and therefore remains an interesting direction for future work.

The prior work \cite{xu2024bit} studied real-valued Toeplitz covariance estimation from incomplete quantized data, and established corresponding non-asymptotic error bounds, within the triangular-dithering quantization framework introduced in \cite{chen2023quantizing}. 
However, extending these results to the complex domain is nontrivial:
one must simultaneously preserve the Hermitian Toeplitz structure and account for the coupling between the real and imaginary parts.
A naive real-lifting approach that stacks the real and imaginary components into a $2d$-dimensional real vector formally reduces the problem to a real-valued covariance estimation problem \cite{mahot2013asymptotic}. 
However, this lifting represents the original Hermitian Toeplitz covariance matrix as a real-valued covariance matrix with a structured $2\times2$ block-Toeplitz form, whose blocks are coupled through the real and imaginary parts of the same Hermitian Toeplitz covariance matrix.
Consequently, the Toeplitz-prior-based theory developed in \cite{xu2024bit} is not directly applicable.

In this paper, we propose a computationally efficient estimator, termed the quantized Toeplitz-projected sample covariance matrix (Q-TSCM), for Hermitian Toeplitz covariance estimation, and analyze its estimation performance under the complex Gaussian assumption.
In addition, we develop an algorithmically refined estimator that improves accuracy by further exploiting the positive semidefinite (PSD) property.
The contributions of this paper are summarized as follows:
\begin{enumerate}
    \item In sparse-observation scenarios, we extend the quantized Toeplitz covariance estimation framework of~\cite{xu2024bit} to the complex domain and propose an unbiased estimator, termed Q-TSCM, along with its finite-bit counterpart, the $2k$-bit Toeplitz-projected sample covariance matrix ($2k$-TSCM). Both estimators are applicable to ruler-based subsampling schemes.
    \item Under the complex Gaussian assumption, we derive non-asymptotic error bounds for the proposed estimators. The bounds explicitly characterize the quadratic dependence on the quantization levels for the real and imaginary parts, and quantify the impact of the observation ruler via the concept of coverage coefficient.
    \item Based on a covariance-fitting criterion, we propose the quantized sparse and parametric approach (Q-SPA) as a quantized analog of SPA. By solving a semidefinite program, Q-SPA further improves covariance estimation accuracy by simultaneously enforcing the Hermitian Toeplitz structure and positive semidefiniteness of the covariance matrix.
    \item We also apply the proposed estimators to DOA estimation with sparse arrays, further illustrating their practical utility.
\end{enumerate}

The remainder of this paper is organized as follows.
Section~\ref{Se_Preliminaries} briefly introduces the notation and preliminaries required for the subsequent analysis.
In Section~\ref{Se_Q_TSCM}, we propose the Q-TSCM and derive non-asymptotic upper bounds for its estimation error.
Building upon this, Section~\ref{Se_2k_TSCM} addresses scenarios involving finite-bit quantized data, introduces the $2k$-TSCM, and provides the corresponding performance analysis.
To further enhance estimation performance, Section~\ref{Se_Algorithm} presents the Q-SPA, an algorithmic refinement based on the covariance fitting criterion.
Finally, Section~\ref{Se_Experiments} reports numerical experiments that validate the proposed estimators, and demonstrates their application to DOA estimation.

\textit{Notation:}
We adopt the following notation throughout the paper.
The sets of real and complex numbers are denoted by $\mathbb{R}$ and $\mathbb{C}$, respectively.
For a positive integer $d$ we write $[d]=\{0, 1,\dots,d-1\}$.
For vectors and matrices, we use $\|\cdot\|_2$ to denote the Euclidean norm (for vectors) or the operator norm (for matrices), and $\|\cdot\|_F$ to denote the Frobenius norm.
We write $\boldsymbol{A} \geq 0$ if all entries of $\boldsymbol{A}$ are nonnegative, and $\boldsymbol{A} \succeq \boldsymbol{B}$ if $\boldsymbol{A}-\boldsymbol{B}$ is positive semidefinite.
For quantities $x$ and $y$, we write $x \lesssim y$ if $x \le C y$ for some universal constant $C>0$, and conversely $x \gtrsim y$ if $x \ge C y$.
We write $x \asymp y$ if both $x \lesssim y$ and $x \gtrsim y$ hold simultaneously.

We denote by $\mathcal{T}(\boldsymbol{\gamma})$ the Hermitian Toeplitz matrix generated by the sequence $\boldsymbol{\gamma}$.
Sparse observation locations are specified by a ruler $\Omega \subset \{1,2,\ldots,d\}$.
The quantization level is collected in a vector $\boldsymbol{\Delta}$, and $\mathcal{Q}_{\boldsymbol{\Delta}}(\cdot)$ denotes the element-wise quantization operator with quantization level $\boldsymbol{\Delta}$.
Symbols with a hat (e.g., $\widehat{\boldsymbol{T}}$, $\widehat{\boldsymbol{T}}_{2k}$) denote covariance estimators constructed from quantized data, while symbols with a dot (e.g., $\dot{\boldsymbol{z}}$) denote the corresponding quantized observations.

Throughout the paper, $c$ and $C$ denote positive universal constants whose values may change from line to line, unless otherwise specified.

\section{Preliminaries} \label{Se_Preliminaries}

\subsection{Dithered Quantization Scheme}

For a complex signal $\boldsymbol{z} \in \mathbb{C}^d$, the dithered quantization is defined by
\begin{equation}
    \label{Eq_Complex_Quantization}
    \dot{\boldsymbol{z}} = \mathcal{Q}_{\Delta_r, \Delta_i} (\boldsymbol{z} + \boldsymbol{\tau}),
\end{equation}
where $\boldsymbol{\tau} \in \mathbb{C}^d$ is a random dither vector, independent of $\boldsymbol{z}$, whose entries are i.i.d. samples drawn from a prescribed distribution.
The operator $\mathcal{Q}_{\Delta_r, \Delta_i} (\cdot)$ denotes entry-wise quantization of complex-valued data and is given by
\begin{equation}
    \mathcal{Q}_{\Delta_r, \Delta_i} (z) := \mathcal{Q}_{\Delta_r} \left(\Re(z)\right) + i \times \mathcal{Q}_{\Delta_i} \left(\Im(z)\right).
\end{equation}
For a real number $x$, the memoryless scalar quantizer $\mathcal{Q}_{\Delta}(\cdot)$ is defined by
\begin{equation}
    \label{Eq_Quantization_Real}
    \mathcal{Q}_{\Delta}(x) = \Delta \left( \left\lfloor \dfrac{x}{\Delta} \right\rfloor + \dfrac{1}{2}\right).
\end{equation}
To ensure that \eqref{Eq_Quantization_Real} is well defined for $\Delta = 0$, we additionally set $\mathcal{Q}_{0}(x) = x$.
Clearly, this quantization scheme generalizes the real-valued setting considered in \cite{xu2024bit, chen2023quantizing, chen2025parameter} to the complex setting. 
We refer to $\boldsymbol{\Delta} = (\Delta_r,\Delta_i)$ as the \emph{quantization level}, where larger components of $\boldsymbol{\Delta}$ correspond to a lower quantization resolution.

In practice, applying an appropriate dither before quantization is of paramount importance.
On the one hand, without dithering, different original variances can collapse to the same quantized variance, thereby hindering the recovery of the true covariance matrix from the quantized data; see \cite{xu2024bit} for a more detailed discussion.
On the other hand, incorporating a properly designed dither can endow the quantized measurements with desirable statistical properties for covariance estimation.
To illustrate this, we state the following lemma, which specializes the results of \cite{gray2002dithered} to our setting.

\begin{lemma}
    \label{Lem_quantization_dither}
    Let $\boldsymbol{x}\in\mathbb{R}^d$ be a random real-valued signal, and assume that the dither vector $\boldsymbol{\tau}$ is independent of $\boldsymbol{x}$. 
    For a prescribed quantization level $\Delta>0$, consider the \emph{triangular quantization}, i.e., the entries of the dither vector $\boldsymbol{\tau}$ are independent and given by
    \begin{equation}
        \label{Eq_triangular_distribution}
        \tau_j \sim \mathcal{U} \left(\left[-\dfrac{\Delta}{2}, \dfrac{\Delta}{2}\right]\right) + \mathcal{U} \left(\left[-\dfrac{\Delta}{2}, \dfrac{\Delta}{2}\right]\right).
    \end{equation}
    Let $\dot{\boldsymbol{x}} = \mathcal{Q}_{\Delta}(\boldsymbol{x} + \boldsymbol{\tau})$.
    Define the quantization error $\boldsymbol{\omega} = \dot{\boldsymbol{x}} - (\boldsymbol{x} + \boldsymbol{\tau})$ and the quantization noise $\boldsymbol{\xi} = \dot{\boldsymbol{x}} - \boldsymbol{x}$.
    Then:
    \begin{itemize}
        \item[a)] $\boldsymbol{\omega}$ is independent of $\boldsymbol{x}$, and its entries are i.i.d. with distribution $\mathcal{U} ([-\frac{\Delta}{2}, \frac{\Delta}{2}])$.
        \item[b)] the conditional second-order moment of $\boldsymbol{\xi}$ given $\boldsymbol{x}$ does not depend\footnote{In general, $\boldsymbol{\xi}$ and $\boldsymbol{x}$ are not independent.} on the value of $\boldsymbol{x}$ and satisfies
        \begin{equation}
            \mathbb{E} [\boldsymbol{\xi} \boldsymbol{\xi}^T | \boldsymbol{x}] = \frac{\Delta^2}{4} \boldsymbol{I}_d,
        \end{equation}
        where $\boldsymbol{I}_d$ denotes the $d \times d$ identity matrix.
    \end{itemize}
\end{lemma}

It follows directly from Lemma~\ref{Lem_quantization_dither} that
\begin{equation}
    \mathbb{E}[\boldsymbol{\xi}| \boldsymbol{x}] = \mathbb{E}[\boldsymbol{\omega} + \boldsymbol{\tau}| \boldsymbol{x}] = \mathbb{E}[\boldsymbol{\omega}] + \mathbb{E}[\boldsymbol{\tau}] = \boldsymbol{0}.
\end{equation}
Consequently,
\begin{equation}
    \mathbb{E}[\boldsymbol{\xi}] = \mathbb{E}\big[\mathbb{E}[\boldsymbol{\xi}| \boldsymbol{x}]\big] = \boldsymbol{0}.
\end{equation}
Moreover, using the conditional second-order moment in Lemma~\ref{Lem_quantization_dither}, we have
\begin{equation}
    \mathbb{E} [\boldsymbol{\xi} \boldsymbol{\xi}^T]
    = \mathbb{E} \big[\mathbb{E} [\boldsymbol{\xi} \boldsymbol{\xi}^T | \boldsymbol{x}]\big]
    = \frac{\Delta^2}{4} \boldsymbol{I}_d.
\end{equation}
Since $\dot{\boldsymbol{x}}=\boldsymbol{x}+\boldsymbol{\xi}$, it follows that \cite{chen2023quantizing, xu2024bit}
\begin{equation}
    \label{Eq_E_real}
    \begin{aligned}
        \mathbb{E}\left[\dot{\boldsymbol{x}}\dot{\boldsymbol{x}}^T\right]
        =& \mathbb{E}\left[(\boldsymbol{x}+\boldsymbol{\xi})(\boldsymbol{x}+\boldsymbol{\xi})^T\right] \\
        =& \mathbb{E}\left[\boldsymbol{x}\boldsymbol{x}^T\right] + \mathbb{E}\left[\boldsymbol{x}\mathbb{E}[\boldsymbol{\xi}^T\mid\boldsymbol{x}]\right] \\
        &+ \mathbb{E}\left[\mathbb{E}[\boldsymbol{\xi}\mid\boldsymbol{x}]\boldsymbol{x}^T\right] + \mathbb{E}\left[\mathbb{E}[\boldsymbol{\xi}\boldsymbol{\xi}^T\mid\boldsymbol{x}]\right] \\
        =& \mathbb{E}\left[\boldsymbol{x}\boldsymbol{x}^T\right] + \frac{\Delta^2}{4}\boldsymbol{I}_d.
    \end{aligned}
\end{equation}
It makes triangular quantization a natural choice for quantized covariance estimation, which has been adopted in \cite{xu2024bit, chen2023quantizing, chen2025parameter, chen2023quantized, dirksen2025subspace}.

When the underlying signal is real-valued, say $z=x$, the complex quantization scheme \eqref{Eq_Complex_Quantization} reduces to the real-valued one by setting the imaginary quantization level to $\Delta_i=0$. 
Indeed, under triangular dithering, $\Delta_i=0$ makes the imaginary dither degenerate at zero, i.e., $\tau_i=0$, and hence $\Im(\dot z)=0$. 
Consequently, $\dot z=\mathcal{Q}_{\Delta_r}(x+\tau_r)$, which is exactly the real-valued triangular quantization model considered in \cite{xu2024bit}. 
In contrast, choosing $\Delta_i>0$ for a real-valued signal would introduce an artificial imaginary dither-and-quantization component, and therefore would no longer correspond to the real-valued setting in \cite{xu2024bit}.


\subsection{Ruler-based Toeplitz Covariance Estimation}
In classical covariance estimation, we typically approximate $\boldsymbol{\Sigma}$ by the sample covariance matrix $\frac{1}{n}\sum_{l=1}^n \boldsymbol{z}^{(l)} \boldsymbol{z}^{(l) H}$ due to its simplicity and computational efficiency.
However, in many practical scenarios we only have access to incomplete observations, i.e., we observe only the subsample $\boldsymbol{z}_\Omega^{(l)}$ over a subset of indices $\Omega$, and are therefore restricted to estimating the corresponding submatrix $\boldsymbol{\Sigma}_\Omega$ of $\boldsymbol{\Sigma}$.
In such cases, structural information can be leveraged to assist the estimation \cite{aubry2021structured}.
In particular, if $\boldsymbol{\Sigma}$ is Toeplitz, then for a suitable index set $\Omega$ (see Definition~\ref{Def_Ruler}), the fact that \textit{all entries along each diagonal of a Toeplitz matrix are equal} allows recovery of the full covariance matrix from its observed submatrix.

A Hermitian Toeplitz matrix $\boldsymbol{T} \in \mathbb{C}^{d\times d}$ can be parameterized by its generators as
\begin{equation}
    \boldsymbol{T} = \mathcal{T}(\boldsymbol{\gamma}) =
    \begin{bmatrix}
    \gamma_0 & \gamma_{1} & \gamma_{2} & \cdots & \gamma_{d-1} \\
    \gamma_{-1} & \gamma_0 & \gamma_{1} & \cdots & \gamma_{d-2} \\
    \gamma_{-2} & \gamma_{-1} & \gamma_0 & \cdots & \gamma_{d-3} \\
    \vdots & \vdots & \vdots & \ddots & \vdots \\
    \gamma_{-d+1} & \gamma_{-d+2} & \gamma_{-d+3} & \cdots & \gamma_0 \\
    \end{bmatrix},
\end{equation}
where $\gamma_{-s} = \gamma_s^*$ for $s \in [d]$.
Under this parametrization, estimating the Toeplitz covariance $\boldsymbol{T}$ is equivalent to estimating its generators $\{\gamma_0, \gamma_{1}, \ldots, \gamma_{d-1}\}$.
Mathematically, for each $s \in [d]$, $\gamma_s$ represents the covariance between two measurements separated by $s$ indices (i.e., at lag $s$).
Accordingly, an estimator of $\gamma_s$ can be constructed as
\begin{equation}
    \label{Eq_estimation_of_gamma_s}
    \tilde{\gamma}_s = \dfrac{1}{n \rvert\Omega_s\rvert}\sum_{l=1}^{n} \sum_{(j, k) \in \Omega_s} z_j^{(l)} z_k^{(l)^*},
\end{equation}
where $\Omega_s = \{(j, k) \in \Omega \times \Omega : k - j = s\}$\footnote{The definition of $\Omega_s$ slightly differs from \cite{eldar2020sample, xu2024bit} in which we enforce an ordering on index pairs $(j, k)$.
This is essential in the complex-valued setting: if one symmetrically includes both $(j, k)$ and $(k, j)$ for the same lag, the estimator effectively averages $\gamma_s$ and $\gamma_{-s} = \gamma_s^*$, yielding $\Re(\gamma_s)$ and discarding the imaginary part.}.

To obtain all the generators $\tilde{\gamma}_s$ for $s \in [d]$, we must ensure that each set $\Omega_s$ is nonempty.
This motivates the following notion of a ruler \cite{eldar2020sample}.

\begin{definition}
    \label{Def_Ruler}
    A subset $\Omega$ is called a \textit{ruler} if for any $s \in [d]$, there exist indices $j, k\in \Omega$ such that $k - j = s$.
\end{definition}

A ruler is thus a set of (sparsely) located observation indices.
When $\lvert \Omega \rvert < d$, such a configuration is also known in array signal processing as a redundant sparse array (RSA) \cite{wu2016direction, yang2014discretization}.
For $\alpha \in [1/2, 1]$, a special class of rulers $\Omega_\alpha$ is defined by
\begin{equation}
    \Omega_\alpha = \Omega_\alpha^{(1)} \cup \Omega_\alpha^{(2)},
\end{equation}
where\footnote{For simplicity, we assume that both $d^\alpha$ and $d^{1-\alpha}$ are integers; otherwise, rounding them to the nearest integer.}
\begin{equation}
    \begin{aligned}
        \Omega_\alpha^{(1)} &= \{1, 2, \ldots, d^\alpha\},\\
        \Omega_\alpha^{(2)} &= \{d, d - d^{1-\alpha}, \ldots, d - (d^\alpha - 1) d^{1-\alpha}\}.
    \end{aligned}
\end{equation}
In particular, $\Omega_1 = \{1, 2, \ldots, d\}$ corresponds to the fully observed (complete) case.
To assess the impact of different rulers on covariance estimation, we introduce the \emph{coverage coefficient}
\begin{equation}
    \phi(\Omega) = \sum_{s=0}^{d-1} \dfrac{1}{\vert \Omega_s\vert}.
\end{equation}
Intuitively, the more frequently a distance $s$ is represented in the ruler $\Omega$, the smaller the value of $\phi(\Omega)$ becomes.
Especially, it is known that\cite{xu2024bit, eldar2020sample}
\begin{equation}
    \label{Eq_phi_Omega_121}
    \phi(\Omega_1) = \mathcal{O}(\log d),\quad \phi(\Omega_{1/2}) = \mathcal{O}(d).
\end{equation}
More generally, the value of $\phi(\Omega)$ depends on several factors, including the dimension $d$, the cardinality $\lvert \Omega \rvert$, and the positional structure of the indices in $\Omega$.

Given a fixed ruler $\Omega$, we may use \eqref{Eq_estimation_of_gamma_s} to estimate the coefficients $\tilde{\gamma}_s$, thereby obtaining the unbiased and consistent Toeplitz covariance estimator
\begin{equation*}
    \tilde{\boldsymbol{T}} = \mathcal{T}(\tilde{\gamma}_{0}, \tilde{\gamma}_{1}, \ldots, \tilde{\gamma}_{d-1}).
\end{equation*}
for $\boldsymbol{T}$.
The resulting \textit{Toeplitz-projected sample covariance matrix (TSCM)} $\tilde{\boldsymbol{T}}$ is computationally simple and requires no iterative algorithm.
Its non-asymptotic performance has been analyzed in detail in \cite{eldar2020sample, qiao2017gridless}.

\section{Quantized Toeplitz Covariance Estimation} \label{Se_Q_TSCM}
We are interested in estimating the true covariance matrix $\boldsymbol{T}$ from partially observed, coarsely quantized observations $\dot{\boldsymbol{z}}_\Omega$.
Specifically, we assume that the covariance matrix $\boldsymbol{T}$ associated with the distribution $\mathcal{D}$ is Hermitian, positive semidefinite, and Toeplitz, and that $\boldsymbol{z} \sim \mathcal{D}$ is a $d$-dimensional complex-valued random vector.
Let $\boldsymbol{z}^{(1)}, \boldsymbol{z}^{(2)}, \ldots, \boldsymbol{z}^{(n)}$ be i.i.d.\ copies of $\boldsymbol{z}$.
For a prescribed quantization level $\boldsymbol{\Delta} \geq 0$ and a given ruler $\Omega$, we observe the quantized samples $\dot{\boldsymbol{z}}_\Omega^{(l)}$, $l = 1, 2, \ldots, n$.
Our objective is to accurately estimate the Toeplitz covariance $\boldsymbol{T}$ based solely on these incomplete quantized observations.

\subsection{The Q-TSCM Estimator}
In this paper, we focus on the case when \textit{triangular quantization} is employed.
That is, the dither in \eqref{Eq_Complex_Quantization} is given by $\boldsymbol{\tau} = \boldsymbol{\tau}_r + i \boldsymbol{\tau}_i$, where the real and imaginary parts $\boldsymbol{\tau}_r$ and $\boldsymbol{\tau}_i$ have i.i.d.\ entries following the triangular distribution specified in \eqref{Eq_triangular_distribution}.
We make the following assumption.
\begin{assumption}
    \label{As_A1}
    Given a quantization level $\boldsymbol{\Delta} \geq 0$ and a ruler $\Omega$, we employ triangular quantization as described above.
    The estimators $\widehat{\gamma}_s$ and $\widehat{\boldsymbol{T}}$ are constructed solely from the incomplete quantized observations $\dot{\boldsymbol{z}}_\Omega$.
\end{assumption}

Let $z_j^{(l)} = x_j^{(l)} + i y_j^{(l)}$ and $\dot{z}_j^{(l)} = \dot{x}_j^{(l)} + i \dot{y}_j^{(l)}$ denote the original and quantized signals, respectively.
Invoking \eqref{Eq_E_real}, we obtain
\begin{equation}
    \label{Eq_E_complex}
    \begin{aligned}
        &\mathbb{E} \left[ \dot{z}_j^{(l)} \dot{z}_k^{(l)^*} \right]\\
        =& \mathbb{E} \left[ \left(\dot{x}_j^{(l)} + i \dot{y}_j^{(l)}\right) \left(\dot{x}_k^{(l)} - i \dot{y}_k^{(l)}\right) \right]\\
        =& \mathbb{E} \left[\dot{x}_j^{(l)} \dot{x}_k^{(l)} \right] + \mathbb{E} \left[\dot{y}_j^{(l)} \dot{y}_k^{(l)} \right] + i \left(\mathbb{E} \left[\dot{y}_j^{(l)} \dot{x}_k^{(l)} \right] - \mathbb{E} \left[\dot{x}_j^{(l)} \dot{y}_k^{(l)} \right]\right)\\
        =& \mathbb{E} \left[x_j^{(l)} x_k^{(l)} \right] + \dfrac{\Delta_r^2}{4} \delta_{j,k} + \mathbb{E} \left[y_j^{(l)} y_k^{(l)} \right] + \dfrac{\Delta_i^2}{4} \delta_{j,k} \\
        &+ i \left( \mathbb{E} \left[y_j^{(l)} x_k^{(l)} \right] - \mathbb{E} \left[x_j^{(l)} y_k^{(l)} \right] \right)\\
        =& \mathbb{E} \left[ \left(x_j^{(l)} + i y_j^{(l)}\right) \left(x_k^{(l)} - i y_k^{(l)}\right) \right] + \dfrac{\Delta_i^2 + \Delta_r^2}{4} \delta_{j,k}\\
        =& \mathbb{E} \left[z_j^{(l)} z_k^{(l)^*} \right] + \dfrac{\|\boldsymbol{\Delta}\|_2^2}{4} \delta_{j,k},
    \end{aligned}
\end{equation}
where $\delta_{j,k}$ denotes the Kronecker delta.
Equation~\eqref{Eq_E_complex} shows that the difference between the second-order moment of the quantized data and that of the full-digital data is exactly the additive term $\frac{\|\boldsymbol{\Delta}\|_2^2}{4}\,\delta_{j,k}$, arising from quantization.

Combining \eqref{Eq_estimation_of_gamma_s} with the identity in \eqref{Eq_E_complex}, we define the per-sample estimator
\begin{equation}
    \widehat{\gamma}_s^{(l)} = \dfrac{1}{|\Omega_s|} \sum_{(j, k) \in \Omega_s} \dot{z}_j^{(l)} \dot{z}_k^{(l)^*} - \dfrac{\|\boldsymbol{\Delta}\|_2^2}{4} \delta_s,\quad s\in [d],
\end{equation}
and the overall estimator
\begin{equation}
    \label{Eq_estimator_gamma}
    \widehat{\gamma}_s = \dfrac{1}{n} \sum_{l = 1}^n \widehat{\gamma}_s^{(l)},\quad s\in [d].
\end{equation}
The corresponding Toeplitz covariance estimator is then defined as
\begin{equation}
    \label{Eq_estimator_T}
    \widehat{\boldsymbol{T}} = \mathcal{T}(\widehat{\gamma}_{0}, \widehat{\gamma}_{1}, \ldots, \widehat{\gamma}_{d-1}).
\end{equation}
We refer to $\widehat{\boldsymbol{T}}$ as the \textit{quantized Toeplitz-projected sample covariance matrix (Q-TSCM)}.
It can be verified that $\widehat{\boldsymbol{T}}$ is an unbiased estimator of $\boldsymbol{T}$, as stated below.

\begin{theorem}
    \label{Th_Unbiased}
    Fix a quantization level $\boldsymbol{\Delta} \geq 0$ and a ruler $\Omega$.
    Under Assumption~\ref{As_A1}, the estimators $\widehat{\gamma}_s$ defined in \eqref{Eq_estimator_gamma} and $\widehat{\boldsymbol{T}}$ defined in \eqref{Eq_estimator_T} are unbiased, i.e.,
    \begin{equation}
        \mathbb{E}[\widehat{\gamma}_s] = \gamma_s, \quad s \in [d], \qquad \text{and} \qquad \mathbb{E}[\widehat{\boldsymbol{T}}] = \boldsymbol{T}.
    \end{equation}
\end{theorem}
\begin{proof}
    It follows directly from \eqref{Eq_E_complex}, together with the definitions of $\Omega_s$, $\widehat{\gamma}_s^{(l)}$.
\end{proof}

Theorem~\ref{Th_Unbiased} shows that, under triangular quantization, one can construct an unbiased estimator of $\boldsymbol{T}$.
This idea has been explored in several recent works; see, e.g., \cite{chen2023quantized, chen2023quantizing, xu2024bit,dirksen2025subspace}.
These developments suggest that triangular quantization provides a principled and convenient framework for multi-bit quantized (Toeplitz) covariance estimation.

\subsection{Non-Asymptotic Performance Analysis}
In the preceding analysis, we have imposed only minimal assumptions on the distribution $\mathcal{D}$, namely that its covariance matrix is Toeplitz.
Consequently, Theorem~\ref{Th_Unbiased} holds for a broad class of underlying distributions.
To analyze the approximation performance of the proposed estimator, we henceforth additionally assume that $\mathcal{D}$ is a zero-mean complex Gaussian distribution\footnote{This complex Gaussian assumption is adopted for analytical convenience and is standard in covariance estimation and signal processing; see, e.g., \cite{dirksen2025subspace, yang2022nonasymptotic}.}, i.e., $\boldsymbol{z} \sim \mathcal{CN}(0, \boldsymbol{T})$.
We summarise this in the following assumption.
\begin{assumption}
    \label{As_A2}
    The underlying samples $\boldsymbol{z}^{(l)}, l=1,2,\ldots, n$ are i.i.d.\ copies of $\boldsymbol{z} \sim \mathcal{CN}(0, \boldsymbol{T})$, where $\boldsymbol{T}$ is a Hermitian positive semidefinite Toeplitz matrix.
\end{assumption}

We first investigate the estimation error of the Toeplitz generators $\gamma_s$, i.e., we bound $\lvert \widehat{\gamma}_s - \gamma_s \rvert$, where the estimator $\widehat{\gamma}_s$ is defined in \eqref{Eq_estimator_gamma}.
The following theorem provides a non-asymptotic high-probability error bound.

\begin{theorem}
    \label{Th_bound_gamma}
    Fix a quantization level $\boldsymbol{\Delta} \geq 0$ and a ruler $\Omega$.
    Under Assumptions~\ref{As_A1} and~\ref{As_A2}, there exist universal constants $c, C > 0$ such that, for every $s \in [d]$ and $t > 0$,
    \begin{equation}
        \label{Eq_Bound_For_Gamma}
        \mathbb{P}(|\widehat{\gamma}_s - \gamma_s| > t) \leq 8 \exp{\left( -c n \min \left\{\dfrac{t^2}{C \mathcal{K}_{\boldsymbol{\Delta}, \boldsymbol{T}}^4}, 1\right\} \right)},
    \end{equation}
    where
    \begin{equation}
        \label{Eq_K}
        \mathcal{K}_{\boldsymbol{\Delta}, \boldsymbol{T}} = \|\boldsymbol{T}\|_2^{1/2} + 2 \| \boldsymbol{\Delta}\|_2
    \end{equation}
    is a constant determined by the covariance $\boldsymbol{T}$ and the quantization level $\boldsymbol{\Delta}$.
    In particular, for any sufficiently large $\delta$ with $\delta \lesssim n$, it holds with probability at least $1 - e^{-\delta}$ that
    \begin{equation}
        \label{Eq_error_gamma}
        |\widehat{\gamma}_s - \gamma_s| \lesssim \mathcal{K}_{\boldsymbol{\Delta}, \boldsymbol{T}}^2\sqrt{\dfrac{\delta}{n}}.
    \end{equation}
\end{theorem}
\begin{proof}
    See Appendix~\ref{Pr_bound_gamma}.
\end{proof}

We now turn to the estimation error of the Toeplitz covariance $\boldsymbol{T}$, measured by $\lVert \widehat{\boldsymbol{T}} - \boldsymbol{T} \rVert_2$.
As mentioned earlier, once the estimators $\widehat{\gamma}_s$ of the Toeplitz generators $\gamma_s$ are obtained, an unbiased estimator $\widehat{\boldsymbol{T}}$ of $\boldsymbol{T}$ can be directly constructed via \eqref{Eq_estimator_T}.
However, directly analyzing the error $\lVert \widehat{\boldsymbol{T}} - \boldsymbol{T} \rVert_2$ is challenging, since one must simultaneously characterize the impact of data missingness.
In fact, although a related result has been established in the real-valued setting in \cite{xu2024bit}, extending it to the complex-valued case presents substantial difficulties.
A conventional approach to extending results from the real-valued to the complex-valued setting is to represent the complex sample $\boldsymbol{z}$ as $\overline{\boldsymbol{z}} = \bigl(\Re(\boldsymbol{z})^T, \Im(\boldsymbol{z})^T\bigr)^T$.
Under the proper complex Gaussian model with covariance $\boldsymbol{\Sigma}$, the lifted real covariance matrix satisfies
\begin{equation}
    \label{Eq_block_toeplitz_lifting}
    \mathbb{E}[\overline{\boldsymbol{z}} \overline{\boldsymbol{z}}^T] = \overline{\boldsymbol{\Sigma}} = \frac{1}{2}
    \begin{pmatrix}
        \Re(\boldsymbol{\Sigma}) & -\Im(\boldsymbol{\Sigma})\\
        \Im(\boldsymbol{\Sigma}) & \Re(\boldsymbol{\Sigma})
    \end{pmatrix},
\end{equation}
from which the complex covariance matrix $\boldsymbol{\Sigma}$ can be recovered; see, e.g., \cite{mahot2013asymptotic}.
This lifting is possible, but the lifted matrix $\overline{\boldsymbol{T}}\in\mathbb{R}^{2d\times2d}$ is a \textit{structured block-Toeplitz} matrix rather than a Toeplitz matrix. 
Therefore, the real-valued Toeplitz results in \cite{xu2024bit} cannot be directly applied to $\overline{\boldsymbol{T}}$.
Without the block-compatibility constraints implicit in \eqref{Eq_block_toeplitz_lifting}, the lifted formulation estimates a larger real covariance model and does not necessarily preserve the desired Hermitian Toeplitz structure.
This motivates us to work directly with the complex Hermitian Toeplitz covariance matrix $\boldsymbol{T}$.


Building on Theorem~\ref{Th_bound_gamma}, we control the spectral density associated with the Toeplitz error matrix $\widehat{\boldsymbol{T}}-\boldsymbol{T}$, and thereby derive the following theorem.
\begin{theorem}
    \label{Th_bound_T}
    Fix a quantization level $\boldsymbol{\Delta} \geq 0$ and a ruler $\Omega$.
    Under Assumptions~\ref{As_A1} and~\ref{As_A2}, there exist universal constants $c, C > 0$ such that, for all $t > 0$,
    \begin{equation}
        \label{Eq_Bound_For_T}
        \begin{aligned}
            \mathbb{P}&\left(\|\widehat{\boldsymbol{T}} - \boldsymbol{T}\|_2 > t\right) \\
            &\leq C d^2 \exp{\left( -c n \min \left\{\dfrac{t^2}{\mathcal{K}_{\boldsymbol{\Delta}, \boldsymbol{T}}^4 \phi(\Omega)}, 1\right\} \right)},
        \end{aligned}
    \end{equation}
    where $\mathcal{K}_{\boldsymbol{\Delta}, \boldsymbol{T}}$ is as defined in \eqref{Eq_K}.
    In particular, for any sufficiently large $\delta > 0$, if $n \gtrsim \delta\log d$, then it holds with probability at least $1 - d^{2-\delta}$ that
    \begin{equation}
        \label{Eq_error_T}
        \|\widehat{\boldsymbol{T}} - \boldsymbol{T}\|_2 \lesssim \mathcal{K}_{\boldsymbol{\Delta}, \boldsymbol{T}}^2 \sqrt{\dfrac{\phi(\Omega) \cdot \delta \log d}{n}}.
    \end{equation}
\end{theorem}
\begin{proof}
    See Appendix~\ref{Pr_bound_T}.
\end{proof}

It follows from Theorem~\ref{Th_bound_gamma} and Theorem~\ref{Th_bound_T} that the effect of quantization on the estimation error is captured through the coefficient $\mathcal{K}_{\boldsymbol{\Delta}, \boldsymbol{T}}$, and hence does not alter the convergence rate in $n$.
Similar observations for \textit{quantized sample covariance matrix (Q-SCM)} have also been reported in, e.g., \cite{chen2023quantizing}.
Moreover, \eqref{Eq_error_T} shows that the impact of quantization on the error $\|\widehat{\boldsymbol{T}} - \boldsymbol{T}\|_2$ is captured through $\mathcal K_{\boldsymbol\Delta,\boldsymbol T}^2$. 
When the quantization-dependent term dominates, the leading dependence on the quantization level is of order $\mathcal{O}(\|\boldsymbol{\Delta}\|_2^2)$. 
This suggests that increasing the quantization level $\boldsymbol{\Delta}$ within a moderate range affects the magnitude of the error bound in a controlled manner, without changing the convergence rate in $n$.
In addition, it also shows that the quantization affects the real and imaginary parts in a symmetric manner though the quantization levels for the two parts can be different.
These behaviours will be verified by our numerical experiments in Section~\ref{Se_Experiments}.

\begin{remark}
    \label{Re_diff_ruler}
    Under sparse observations, the ruler influences the bound on $\|\widehat{\boldsymbol{T}}-\boldsymbol{T}\|_2$ only through $\phi(\Omega)$, which depends not only on $|\Omega|$ but also on the spatial distribution of the observed indices.
    Consequently, rulers with the same cardinality may yield different error bounds.
    For instance, when $d=16$, the rulers
    \begin{equation}
        \label{Eq_Omega_A_B}
        \begin{aligned}
            \Omega_A &= \{ 1,2,3,4,5,6,7,8,16 \},\\
            \Omega_B &= \{ 1,2,3,5,8,11,14,15,16 \}.
        \end{aligned}
    \end{equation}
    both have $|\Omega_A|=|\Omega_B|=9$, yet $\phi(\Omega_A)\approx 10.70$ whereas $\phi(\Omega_B)\approx 7.11$, indicating potentially different estimation performance.
\end{remark}

\subsection{Non-Asymptotic Analysis in Case of Ruler $\Omega_\alpha$}
The preceding results hold for arbitrary rulers $\Omega$.
However, the error bounds in Theorem~\ref{Th_bound_T} can be less transparent, since the upper bound \eqref{Eq_error_T} involves $\phi(\Omega)$, which is generally difficult to compute in closed form for a given ruler.
To better illustrate Theorem~\ref{Th_bound_T}, we next consider the family $\Omega_\alpha$, $\alpha \in [1/2,1]$, as a concrete example.
Before proceeding, we introduce the following lemma \cite{eldar2020sample}.

\begin{lemma}
    \label{La_phi_Omega}
    It holds that $|\Omega_\alpha| = \mathcal{O}(d^\alpha)$ and
    \begin{equation}
        \phi(\Omega_\alpha) \lesssim d^{2-2\alpha} + d^{1-\alpha} \log d
    \end{equation}
    for any fixed $\alpha \in [1/2,1]$.
\end{lemma}

Lemma~\ref{La_phi_Omega} is consistent with \eqref{Eq_phi_Omega_121}.
Combining Theorem~\ref{Th_bound_T} with Lemma~\ref{La_phi_Omega}, we can immediately derive the following corollary.
\begin{corollary}
    \label{Co_bound_Omega_121}
    Under the setting of Theorem~\ref{Th_bound_T}, for $\alpha \in [1/2, 1]$, if we additionally assume that the particular ruler $\Omega_\alpha$ is used, then it holds that, with probability at least $1 - d^{2-\delta}$,
    \begin{equation}
        \label{Eq_error_T_Omega_alpha}
        \|\widehat{\boldsymbol{T}} - \boldsymbol{T}\|_2 \lesssim \mathcal{K}_{\boldsymbol{\Delta}, \boldsymbol{T}}^2 \sqrt{\dfrac{\max\{d^{2-2\alpha}, d^{1-\alpha} \log d\} \cdot \delta \log d}{n}}.
    \end{equation}
    In particular,
    \begin{enumerate}
        \item if the full ruler $\Omega_1$ is applied\footnote{This corresponds to complete observations.}, then
        \begin{equation}
            \label{Eq_error_T_Omega_1}
            \|\widehat{\boldsymbol{T}} - \boldsymbol{T}\|_2 \lesssim \mathcal{K}_{\boldsymbol{\Delta}, \boldsymbol{T}}^2 \sqrt{\dfrac{\delta (\log d)^2}{n}}.
        \end{equation}

        \item if the sparse ruler $\Omega_{1/2}$ is applied, then
        \begin{equation}
            \label{Eq_error_T_Omega_12}
            \|\widehat{\boldsymbol{T}} - \boldsymbol{T}\|_2 \lesssim \mathcal{K}_{\boldsymbol{\Delta}, \boldsymbol{T}}^2 \sqrt{\dfrac{\delta \cdot d\log d}{n}}.
        \end{equation}
    \end{enumerate}
\end{corollary}
\begin{proof}
    The result follows directly from Theorem~\ref{Th_bound_T} together with the bound in Lemma~\ref{La_phi_Omega}.
\end{proof}

These conclusions are consistent with those in \cite{xu2024bit}, which focuses on the real-valued case.
The main modification is to replace the scalar quantization level $\Delta$ by $\|\boldsymbol{\Delta}\|_2$.
Moreover, the resulting bounds improve upon those obtained for Q-SCM, which typically satisfy
\begin{equation*}
    \|\widehat{\boldsymbol{T}}_\text{Q-SCM} - \boldsymbol{T}\|_2 \lesssim \mathcal{K}_{\boldsymbol{\Delta}, \boldsymbol{T}}^2 \sqrt{\dfrac{d\log d}{n}}
\end{equation*}
under complete observations; see, e.g., \cite{chen2023quantizing}.
This gain stems from exploiting the Toeplitz structure.
For large $d$, the Toeplitz prior substantially enhances the estimation accuracy.
In particular, when using Q-TSCM, observing only the entries indexed by $\Omega_{1/2}$ already yields an error rate comparable to that of Q-SCM based on fully observed data.

\section{Few-bit Quantized Toeplitz Covariance Estimation} \label{Se_2k_TSCM}
The Q-TSCM $\widehat{\boldsymbol{T}}$ analyzed above in fact corresponds to an infinite-bit quantization model, since the quantizer \eqref{Eq_Quantization_Real} (and, accordingly, \eqref{Eq_Complex_Quantization}) is a uniform quantizer with infinitely many output levels for (potentially) unbounded inputs.
From the perspective of data storage and transmission, this is not fully satisfactory.
Therefore, in this section we further consider Toeplitz covariance estimation based on \emph{finite-bit} quantized data.

For unbounded data, a common approach is to truncate the data first \cite{chen2023quantizing}, and then apply quantization using \eqref{Eq_Quantization_Real}. This is equivalent to defining the following $k$-bit quantizer\footnote{In fact, $L$-level quantization is equivalent to setting $k = \log_2 L$.}
\begin{equation}
    \mathcal{Q}_{\Delta, k} (x) :=
    \begin{cases}
        (2^{k-1} - 1/2) \Delta, &\quad \mathrm{if}\  x \geq 2^{k-1} \cdot \Delta,\\
        -(2^{k-1} - 1/2) \Delta, &\quad \mathrm{if}\  x \leq -2^{k-1} \cdot \Delta,\\
        \mathcal{Q}_{\Delta}(x), &\quad \mathrm{otherwise}.
    \end{cases}
\end{equation}
and its corresponding $2k$-bit complex quantizer
\begin{equation}
    \label{Eq_2k_bit_complex_quantizer}
    \mathcal{Q}_{\Delta, 2k} (z) := \mathcal{Q}_{\Delta, k} \left(\Re(z)\right) + i \times \mathcal{Q}_{\Delta, k} \left(\Im(z)\right).
\end{equation}
In \eqref{Eq_2k_bit_complex_quantizer}, we allocate an equal number of bits to quantize the real and imaginary parts and employ identical quantization levels, i.e., we set $\boldsymbol{\Delta} = (\Delta, \Delta)$.
This choice is motivated by the symmetric impact of real and imaginary quantization on covariance estimation, as shown in the preceding section.
More general configurations (with possibly different bit allocations or quantization levels) can be analyzed in a similar manner.
We denote by $\widehat{\boldsymbol{T}}_{2k}$ the estimator obtained by applying \eqref{Eq_estimator_gamma} and \eqref{Eq_estimator_T} to the quantized data generated by \eqref{Eq_2k_bit_complex_quantizer}, and refer to it as the \textit{$2k$-bit Toeplitz-projected sample covariance matrix ($2k$-TSCM)}.

We observe that $\widehat{\boldsymbol{T}}_{2k} = \widehat{\boldsymbol{T}}$ holds whenever
\begin{equation}
    \label{Eq_equality_2k}
    \mathcal{Q}_{\Delta, 2k} (z_j^{(l)} + \tau_j^{(l)}) = \mathcal{Q}_{\Delta, \Delta} (z_j^{(l)} + \tau_j^{(l)})
\end{equation}
for all $j \in \Omega$ and $l = 1,2,\ldots,n$.
In fact, \eqref{Eq_equality_2k} is guaranteed if
\begin{equation}
    \label{Eq_max_z_plus_tau}
    \begin{aligned}
        \Delta &> 2^{1-k} \cdot \max_{j\in \Omega, 1\leq l \leq n} \left\{\left| z_j^{(l)} + \tau_j^{(l)} \right|\right\}\\
        &\geq 2^{1-k} \cdot \max_{j\in \Omega, 1\leq l \leq n} \left\{\left| \Re\left(z_j^{(l)} + \tau_j^{(l)}\right) \right|, \left| \Im\left(z_j^{(l)} + \tau_j^{(l)}\right) \right|\right\}.
    \end{aligned}
\end{equation}
This observation suggests that choosing $\Delta$ sufficiently large in \eqref{Eq_2k_bit_complex_quantizer} ensures that $\widehat{\boldsymbol{T}}_{2k}$ and $\widehat{\boldsymbol{T}}$ coincide with high probability and therefore share the same error bounds.
More precisely, we establish the following theorem.

\begin{theorem}
    \label{Th_bound_2k_bit_T}
    Under the setting of Theorem~\ref{Th_bound_T}, for $k \geq 2$ and any sufficiently large $\delta, \delta' > 0$, if we set $\boldsymbol{\Delta} = (\Delta, \Delta)$ and choose
    \begin{equation}
        \Delta = C \cdot 2^{2-k} \sqrt{\gamma_0 \left(\log\left(n |\Omega|\right) + \delta'\right)},
    \end{equation}
    with $C \geq \frac{1}{2 - \sqrt{2}}$. Then, it holds with probability at least $1 - d^{2-\delta} - e^{-\delta'}$ that
    \begin{equation}
        \label{Eq_bound_2k_bit_T}
        \|\widehat{\boldsymbol{T}}_{2k} - \boldsymbol{T}\|_2 \lesssim \mathcal{K}_{\boldsymbol{\Delta}, \boldsymbol{T}}^2 \sqrt{\dfrac{\phi(\Omega) \cdot \delta \log d}{n}},
    \end{equation}
    where $\mathcal{K}_{\boldsymbol{\Delta}, \boldsymbol{T}} = \|\boldsymbol{T}\|_2^{1/2} + 2\sqrt{2} \Delta$ is the constant defined in \eqref{Eq_K}.
    In particular, if
    \begin{equation}
        n \gtrsim \dfrac{1}{|\Omega|} \exp \left(\dfrac{2^{2k}\|\boldsymbol{T}\|_2}{\gamma_0} - \delta'\right),
    \end{equation}
    then it holds that
    \begin{equation}
        \|\widehat{\boldsymbol{T}}_{2k} - \boldsymbol{T}\|_2 \lesssim \dfrac{\gamma_0 \left(\log\left(n |\Omega|\right) + \delta'\right)}{2^{2k}} \sqrt{\dfrac{\phi(\Omega) \cdot \delta \log d}{n}}
    \end{equation}
    with the same probability.
\end{theorem}
\begin{proof}
    See Appendix~\ref{Pr_bound_2k_bit_T}.
\end{proof}

The above results indicate that, when $n$ is sufficiently large, the estimation error of $\widehat{\boldsymbol{T}}_{2k}$ is essentially governed by the main diagonal entry $\gamma_0$ of $\boldsymbol{T}$, rather than by its spectral norm.
To further elucidate Theorem~\ref{Th_bound_2k_bit_T}, we consider the $4$-bit estimator $\widehat{\boldsymbol{T}}_{4}$ (i.e., $k=2$) under the ruler $\Omega_\alpha$ and provide the following corollary.

\begin{corollary}
    \label{Co_bound_4_bit_Omega_alpha}
    Under the setting of Theorem~\ref{Th_bound_2k_bit_T}, if the ruler $\Omega_\alpha, \alpha \in [1/2, 1]$ is used, then
    \begin{equation}
        \begin{aligned}
            \|\widehat{\boldsymbol{T}}_{4} - \boldsymbol{T}\|_2 \lesssim& \gamma_0 \log\left(n d^\alpha\right) \\
            & \quad \cdot \sqrt{\dfrac{\max\{d^{2-2\alpha}, d^{1-\alpha} \log d\} \cdot \log d}{n}}.
        \end{aligned}
    \end{equation}
    holds with high probability.
    In particular, if the full ruler $\Omega_1$ is applied, then
    \begin{equation}
        \|\widehat{\boldsymbol{T}}_{4} - \boldsymbol{T}\|_2 \lesssim \gamma_0 \log\left(n d\right) \cdot \log d \cdot \sqrt{\dfrac{1}{n}}.
    \end{equation}
\end{corollary}
\begin{proof}
    The result follows directly from Theorem~\ref{Th_bound_2k_bit_T} by combining it with the bound on $\phi(\Omega_\alpha)$ in Lemma~\ref{La_phi_Omega}.
\end{proof}

Under full observation, i.e., set $\alpha = 1$, Corollary~\ref{Co_bound_4_bit_Omega_alpha} yields a tighter covariance estimation bound.
Specifically, \cite[Theorem~4]{chen2025parameter} considers a closely related triangular-quantization-based covariance estimation setting (without Toeplitz-projection), and establishes the bound\footnote{This result is stated for the real-valued case. It is straightforward to be extended to the complex-valued setting by a naive real-lifting approach, up to constant factors.}
\begin{equation*}
    \|\widehat{\boldsymbol{T}}_{2b} - \boldsymbol{T}\|_2 \lesssim \sqrt{\dfrac{d\|\boldsymbol{T}\|_\infty \|\boldsymbol{T}\|_2 \log(nd)^2}{n}}.
\end{equation*}
With respect to the dependence on the dimension $d$, our bound in Corollary~\ref{Co_bound_4_bit_Omega_alpha} is asymptotically tighter, since it involves only polylogarithmic factors in $d$, whereas the above bound scales as $\sqrt{d}$.
Moreover, in regimes where $\|\boldsymbol{T}\|_2 \gg \gamma_0$, i.e., when the energy of $\boldsymbol{T}$ is highly concentrated, Corollary~\ref{Co_bound_4_bit_Omega_alpha} also provides a strictly better dependence on the covariance matrix $\boldsymbol{T}$ itself.
Similar improvements are observed when comparing our results with those in \cite{dirksen2025subspace, dirksen2024tuning}.
These gains are attributable to the exploitation of the Toeplitz structure.

It can be verified that, when the underlying signal is real-valued and the imaginary quantization branch is disabled by setting $\Delta_i=0$, the proposed complex-valued Q-TSCM and $2k$-TSCM reduce to their real-valued counterparts studied in \cite{xu2024bit}. 
Accordingly, the corresponding error bounds also recover the real-valued convergence rates in \cite{xu2024bit}, up to universal constants.

\begin{figure}[tbp]
	\centering
	\includegraphics[scale=0.21]{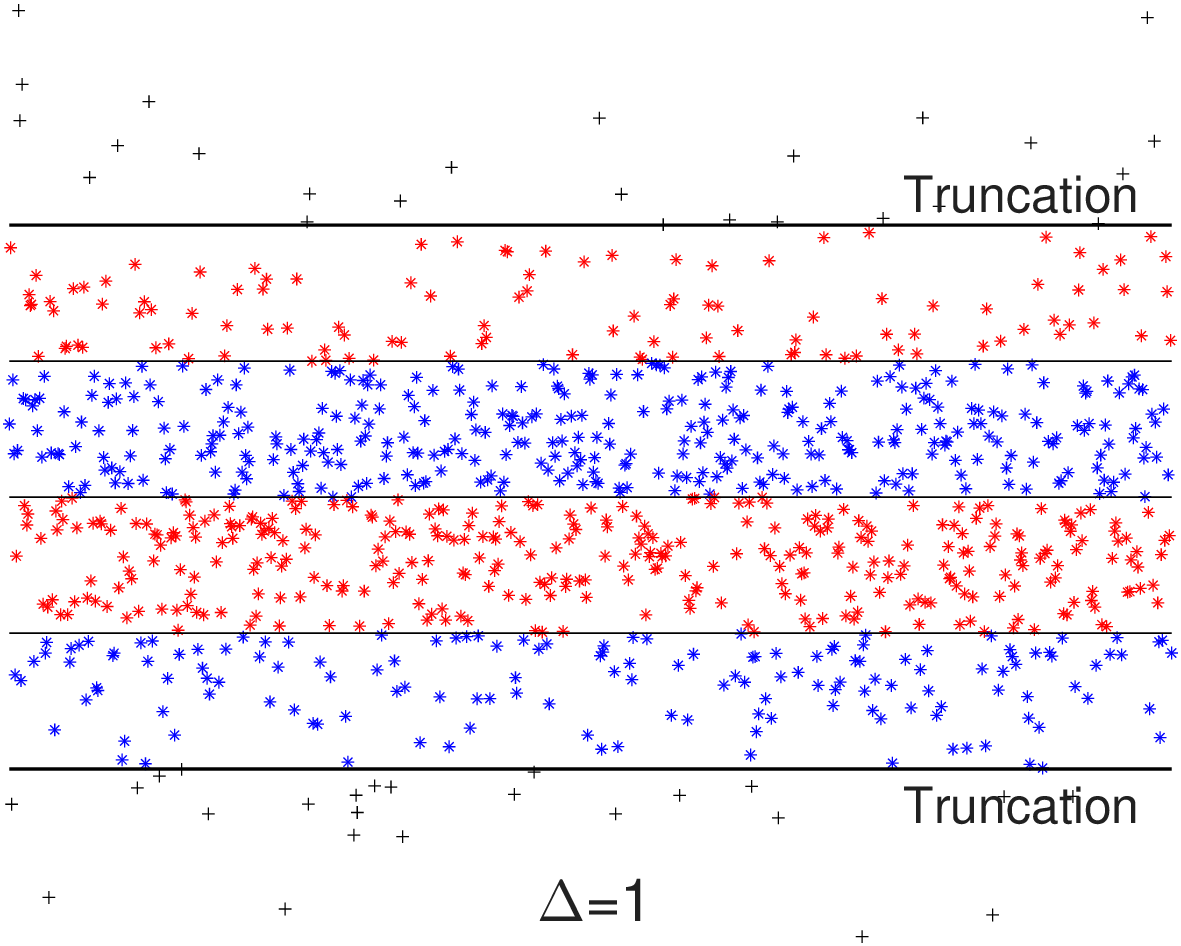}
    \includegraphics[scale=0.21]{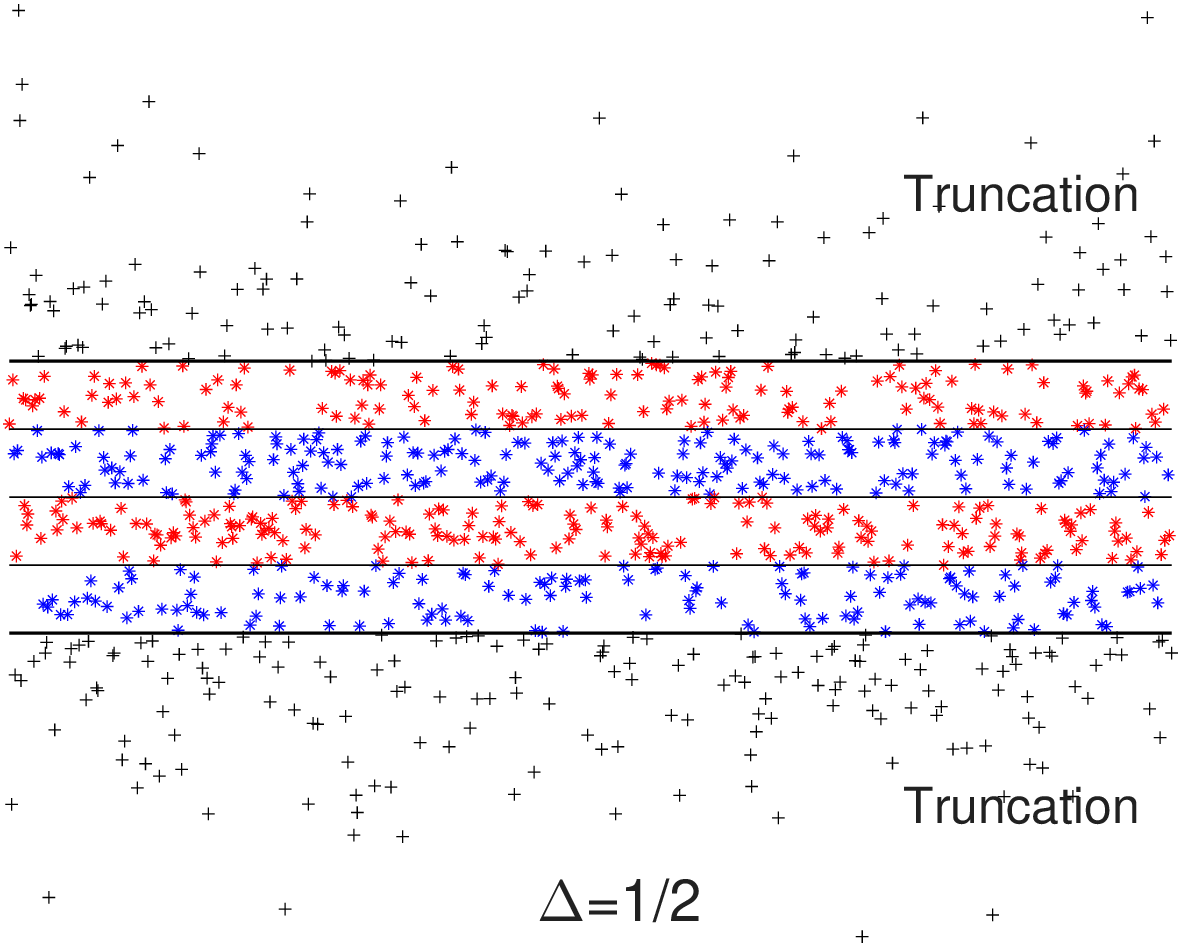}
	\caption{Illustration of the effect of the quantization level $\Delta$ on the quantization of samples drawn from a standard normal distribution, for a fixed bit depth $k=2$.}
	\label{Fig_selection_of_delta}
\end{figure}

In quantized covariance estimation, the quantization level $\Delta$ must be specified a priori.
This choice is nontrivial, since selecting $\Delta$ either too large or too small can substantially increase the estimation error.
As illustrated in Fig.~\ref{Fig_selection_of_delta}, for a fixed bit budget $k$, a large $\Delta$ results in coarse quantization and hence larger quantization error.
Conversely, if $\Delta$ is too small, a substantial fraction of the observations are clipped at the outermost quantization levels, leading to severe information loss and degraded estimation accuracy.
The approach in Theorem~\ref{Th_bound_2k_bit_T} depends on $\gamma_0$, which is typically unknown in practice.
Consequently, developing a parameter-free rule for selecting $\Delta$ merits further investigation.

For general Q-SCM estimators, several parameter-free or data-driven strategies have been proposed to reduce the need for manual tuning. 
For example, \cite{chen2025parameter, dirksen2024tuning} choose the quantization level using information from the unquantized samples, such as
\begin{equation}
    \Delta = s \max_{1 \leq j \leq d, 1 \leq l \leq n} \left|z_j^{(l)}\right|,
\end{equation}
where $s>0$ is a prescribed constant.
However, such a rule relies on unquantized samples and therefore is not directly implementable as a post-processing rule under the strict few-bit observation model considered in this paper. 
It should instead be regarded as an offline quantizer-design heuristic.
Accordingly, in our theoretical analysis, $\Delta$ is treated as a prescribed design parameter.
Identifying a principled quantization-level selection scheme tailored to Toeplitz covariance estimation remains an open problem.


\section{Performance Enhancement via Algorithm} \label{Se_Algorithm}
The Q-TSCM estimator exploits the Toeplitz structure; however, it is not guaranteed to be positive semidefinite in general, which might be of crucial concern in applications.
In the case of unquantized data, existing algorithms, e.g., SPA \cite{yang2014discretization}, have been proposed to utilize the Toeplitz and positive semidefiniteness structures simultaneously to improve the estimation accuracy.
Motivated by SPA, we propose an analogous algorithm, termed the quantized SPA (Q-SPA), for Toeplitz covariance estimation from quantized compressive data.

Under Assumption~\ref{As_A1}, for a prescribed quantization level $\boldsymbol{\Delta} > 0$ and ruler $\Omega$, it follows from Theorem~\ref{Th_Unbiased} that
\begin{equation}
    \mathbb{E}\!\left[ \dot{\boldsymbol{z}}_\Omega \dot{\boldsymbol{z}}_\Omega^H \right] = \dot{\boldsymbol{R}}_\Omega,
\end{equation}
where $\dot{\boldsymbol{R}} = \boldsymbol{T} + \frac{\|\boldsymbol{\Delta}\|_2^2}{4} \boldsymbol{I}_d$ denotes the covariance matrix of the quantized samples, which remains Toeplitz.
We define the sample covariance matrix of the incomplete quantized observations as
\begin{equation}
    \widehat{\boldsymbol{R}}_\Omega = \frac{1}{n} \sum_{l=1}^{n} \dot{\boldsymbol{z}}_\Omega^{(l)} \big(\dot{\boldsymbol{z}}_\Omega^{(l)}\big)^H .
\end{equation}
Adopting the covariance-fitting criterion (metric) used in SPA and accounting for quantization effects, we obtain the following formulation:
\begin{equation}
    \label{Eq_SPA_1}
    \begin{aligned}
        \min_{\boldsymbol{u}}& \ \left\| \dot{\boldsymbol{R}}^{-\frac{1}{2}}_\Omega \big(\widehat{\boldsymbol{R}}_\Omega - \dot{\boldsymbol{R}}_\Omega\big) \widehat{\boldsymbol{R}}^{-\frac{1}{2}}_\Omega \right\|_F^2,\\
        \text{s.t.}& \ \dot{\boldsymbol{R}} = \mathcal{T}(\boldsymbol{u}) \succeq \dfrac{\|\boldsymbol{\Delta}\|_2^2}{4} \boldsymbol{I}_d.
    \end{aligned}
\end{equation}
The objective is adopted from \cite{stoica2010spice} and subsequently used in \cite{yang2014discretization} as a surrogate for the weighted least-squares metric.
The constraint reflects the quantization-induced bias term $\frac{\|\boldsymbol{\Delta}\|_2^2}{4} \boldsymbol{I}_d$ in $\dot{\boldsymbol{R}}$, and ensures the positive semidefiniteness of $\boldsymbol{T}$.

A direct calculation shows that \eqref{Eq_SPA_1} is equivalent to \cite{yang2014discretization}:
\begin{equation}
    \label{Eq_SPA_2}
    \begin{aligned}
        \min_{\boldsymbol{u}}& \ \text{tr}\!\left(\widehat{\boldsymbol{R}}_\Omega^{-1} \dot{\boldsymbol{R}}_\Omega\right) + \text{tr}\!\left(\dot{\boldsymbol{R}}_\Omega^{-1} \widehat{\boldsymbol{R}}_\Omega\right),\\
        \text{s.t.}& \ \dot{\boldsymbol{R}} = \mathcal{T}(\boldsymbol{u}) \succeq \dfrac{\|\boldsymbol{\Delta}\|_2^2}{4} \boldsymbol{I}_d.
    \end{aligned}
\end{equation}
Using the identity
\begin{equation}
    \text{tr}\!\left(\boldsymbol{X}^H \boldsymbol{Y}^{-1} \boldsymbol{X}\right) = \min_{\boldsymbol{U}} \text{tr}(\boldsymbol{U})\quad
    \text{s.t.} \quad \begin{pmatrix}
        \boldsymbol{U} & \boldsymbol{X}^H\\
        \boldsymbol{X} & \boldsymbol{Y}
    \end{pmatrix}
    \succeq \boldsymbol{0},
\end{equation}
the optimization problem \eqref{Eq_SPA_2} can be reformulated as the semidefinite program (SDP)
\begin{equation}
    \label{Eq_SPA_3}
    \begin{aligned}
        \min_{\boldsymbol{u}, \boldsymbol{U}}& \ \text{tr}\!\left(\widehat{\boldsymbol{R}}_\Omega^{-1} \dot{\boldsymbol{R}}_\Omega\right) + \text{tr}(\boldsymbol{U}),\\
        \text{s.t.}& \ \begin{pmatrix}
            \boldsymbol{U} & \widehat{\boldsymbol{R}}_\Omega^{\frac{1}{2}}\\
            \widehat{\boldsymbol{R}}_\Omega^{\frac{1}{2}} & \dot{\boldsymbol{R}}_\Omega
        \end{pmatrix}
        \succeq \boldsymbol{0},\\
        & \ \dot{\boldsymbol{R}} = \mathcal{T}(\boldsymbol{u}) \succeq \dfrac{\|\boldsymbol{\Delta}\|_2^2}{4} \boldsymbol{I}_d,
    \end{aligned}
\end{equation}
which can be solved using standard SDP solvers (e.g., CVX).
Once the solution $\breve{\boldsymbol{u}}$ to $\boldsymbol{u}$ is obtained from \eqref{Eq_SPA_3}, we define
\begin{equation}
    \breve{\boldsymbol{T}} = \mathcal{T}(\breve{\boldsymbol{u}}) - \dfrac{\|\boldsymbol{\Delta}\|_2^2}{4} \boldsymbol{I}_d
\end{equation}
as an estimator of $\boldsymbol{T}$.
We refer to $\breve{\boldsymbol{T}}$ as the Q-SPA estimator, in view of its similarity to SPA in \cite{yang2014discretization}.

It is shown in \cite{yang2018sparse} that SPA can also be interpreted as computing a weighted atomic norm of the square root of the covariance matrix, with the weight determined by the covariance matrix itself.
Therefore, SPA promotes sparse frequency recovery, which corresponds to estimating a low-rank Toeplitz covariance matrix under the Vandermonde decomposition.
Consequently, one may expect Q-SPA to yield improved covariance estimates, especially when the underlying Toeplitz covariance matrix $\boldsymbol{T}$ is (approximately) low-rank.
Numerical results will be provided to validate its improved estimation performance, at the cost of increased computational complexity due to the need to solve an SDP.

We close this section by discussing the computational complexity of the proposed estimators.
For all sample covariance matrix (SCM)-based methods considered here, forming the observed SCM requires $\mathcal{O}(n|\Omega|^2)$ operations. 
After this step, Q-TSCM only performs Toeplitz averaging over entries with the same lag, whose cost is at most $\mathcal{O}(|\Omega|^2+d)$.
In contrast, Q-SPA requires solving the SDP in \eqref{Eq_SPA_3}, and is therefore more computationally demanding.
Following the SDP complexity analysis in \cite{yang2018sparse, vandenberghe1996semidefinite}, the SDP-solving step has worst-case complexity $\mathcal{O}(|\Omega|^4d^{2.5}+d^{4.5})$.
Hence, Q-SPA is mainly intended as a performance-enhancing refinement for small- to moderate-scale problems.
To reduce the computational cost, one may use first-order methods such as the alternating direction method of multipliers (ADMM), whose dominant per-iteration cost is the projection onto the positive semidefinite cones, approximately $\mathcal{O}(|\Omega|^3 + d^{3})$; see \cite[Section 6.8.2]{yang2018sparse} and references therein.
In this paper, we solve Q-SPA via the SDP formulation for simplicity.

\section{Numerical Experiments} \label{Se_Experiments}

We validate the effectiveness of the proposed estimator through numerical experiments, and further demonstrate its practical applicability to DOA estimation.

\subsection{Quantized Toeplitz Covariance Estimation}
We first concentrate on the performance of covariance estimation, analyzing the impact of various factors on the error bounds.
In our numerical experiments, we randomly generate the true covariance matrix $\boldsymbol{T}$.
Specifically, we first sample $d$ distinct frequencies from the uniform distribution $\mathcal{U}(0,1)$ and construct the corresponding Fourier matrix $\boldsymbol{F}_d \in \mathbb{C}^{d \times d}$.
Subsequently, we sample $d$ independent realizations from the standard normal distribution $\mathcal{N}(0,1)$, take their absolute values as amplitudes, and construct the diagonal matrix $D$.
Finally, by exploiting the Vandermonde decomposition of Toeplitz matrices \cite{yang2016vandermonde}, we construct a random Toeplitz covariance matrix $\boldsymbol{T}$.
We then draw $n$ i.i.d. samples $\boldsymbol{z}^{(l)}, l = 1, 2, \ldots, n$, from the complex Gaussian distribution $\mathcal{CN}(\boldsymbol{0}, \boldsymbol{T})$ and obtain their triangularly quantized sub-observations $\dot{\boldsymbol{z}}^{(l)}_\Omega$, which are used for covariance estimation.
Each data point in our plots is set to be the mean value of 100 Monte Carlo simulations.

\begin{figure}[tbp]
	\centering
	\includegraphics[scale=0.28]{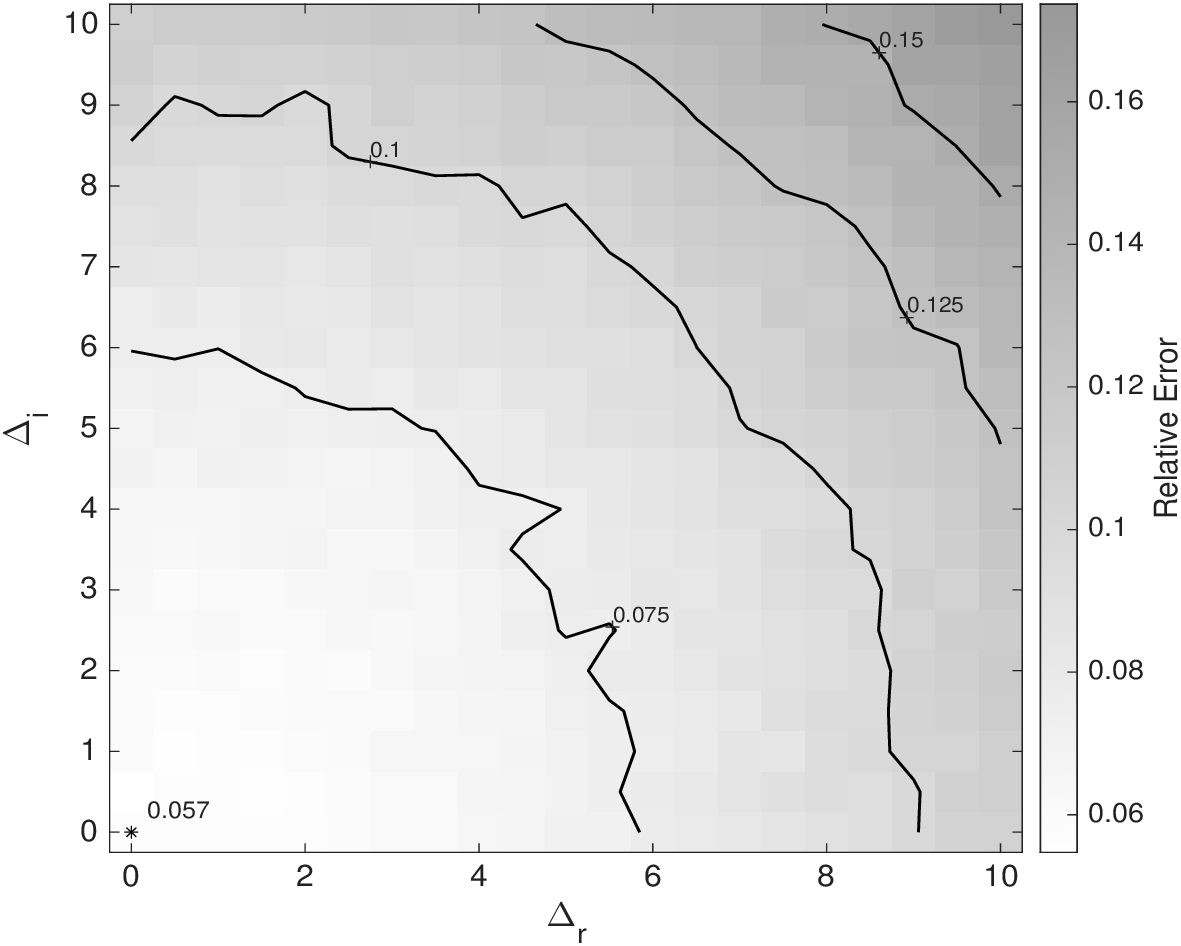}
	\caption{Impact of real and imaginary quantization levels on Toeplitz covariance estimation error.}
	\label{Fig_impact_of_delta}
\end{figure}

In \textit{Experiment 1}, we examine the impact of quantization on the estimation error in both the real and imaginary parts.
Specifically, we set $n=500$, $d=16$, and adopt the full ruler $\Omega_1$.
For varying real and imaginary quantization levels $\Delta_r$ and $\Delta_i$, we compute the relative error and plot the resulting error map with contour lines in Fig.~\ref{Fig_impact_of_delta}.
The contours at different levels are approximately circular arcs, indicating that the effects of real and imaginary quantization are nearly symmetric.
Meanwhile, the contours are relatively sparse near the origin and become denser away from it, consistent with the scaling $\mathcal{O}(\|\boldsymbol{\Delta}\|_2^2)$.
In the subsequent experiments we set $\Delta_r = \Delta_i = \Delta$.

\begin{figure}[tbp]
	\centering
	\includegraphics[scale=0.28]{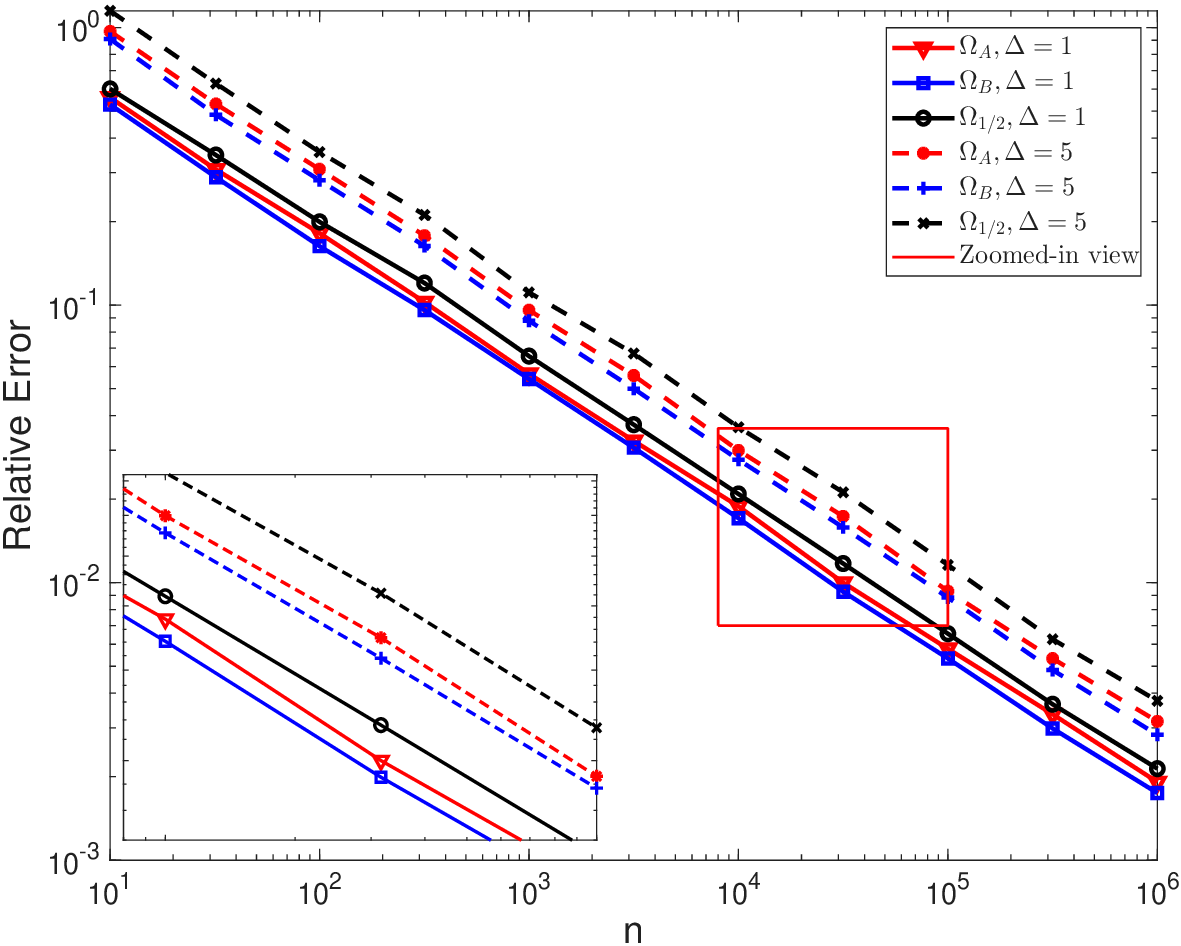}
	\caption{Impact of different rulers on Toeplitz covariance estimation error.}
	\label{Fig_impact_of_different_ruler}
\end{figure}

In \textit{Experiment 2}, we investigate the impact of different rulers. 
For $d=16$, we set $\Delta = 1$ or $5$ and consider three rulers
\begin{equation*}
    \begin{aligned}
        \Omega_A &= \{1, 2, 3, 4, 5, 6, 7, 8, 16\},\\
        \Omega_B &= \{1, 2, 3, 5, 8, 11, 14, 15, 16\},\\
        \Omega_{1/2} &= \{1, 2, 3, 4, 8, 12, 16\},
    \end{aligned}
\end{equation*}
where $|\Omega_{1/2}| = 7 < |\Omega_{A}| = |\Omega_{B}| = 9$, and
\begin{equation*}
    \phi(\Omega_B) < \phi(\Omega_A) < \phi(\Omega_{1/2}).
\end{equation*}
We plot the relative error versus the sample size $n$ on a log-log scale in Fig.~\ref{Fig_impact_of_different_ruler}\footnote{It should be noted that, although sample sizes in signal processing applications rarely approach $10^6$, we intentionally consider a very large range of $n$ to highlight a key distinction: for biased estimators the error typically decreases and then levels off, whereas our (unbiased) estimator continues to improve as $n$ increases.}.
Although $\Omega_A$ and $\Omega_B$ have the same cardinality, estimation based on $\Omega_B$ consistently yields smaller errors, which is consistent with its smaller coverage coefficient, as shown in Remark~\ref{Re_diff_ruler}.
By contrast, estimation based on $\Omega_{1/2}$ exhibits larger errors than the other two rulers, in line with $\phi(\Omega_{1/2})$ being substantially larger than $\phi(\Omega_A)$ and $\phi(\Omega_B)$.
Overall, these results demonstrate that different rulers can lead to markedly different estimation errors, and that this variation is well captured by the corresponding coverage coefficients, consistent with Remark~\ref{Re_diff_ruler}.

\begin{figure}[tbp]
  \centering
  \subfloat[Fix $d = 16$ and vary $n$]{%
    \includegraphics[scale=0.28]{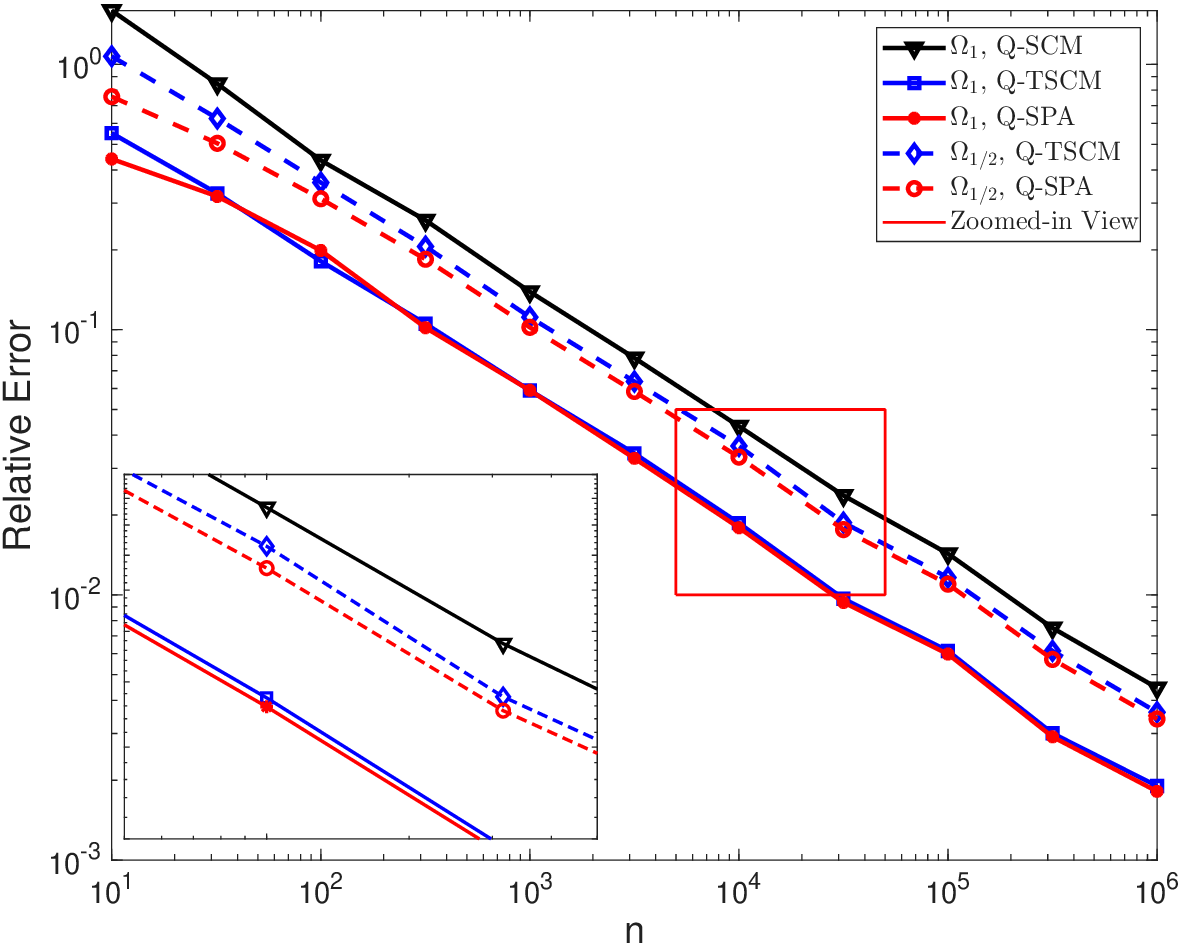}%
    \label{Fig_impact_n}}
  \hfill
  \subfloat[Fix $n = 500$ and vary $d$]{%
    \includegraphics[scale=0.28]{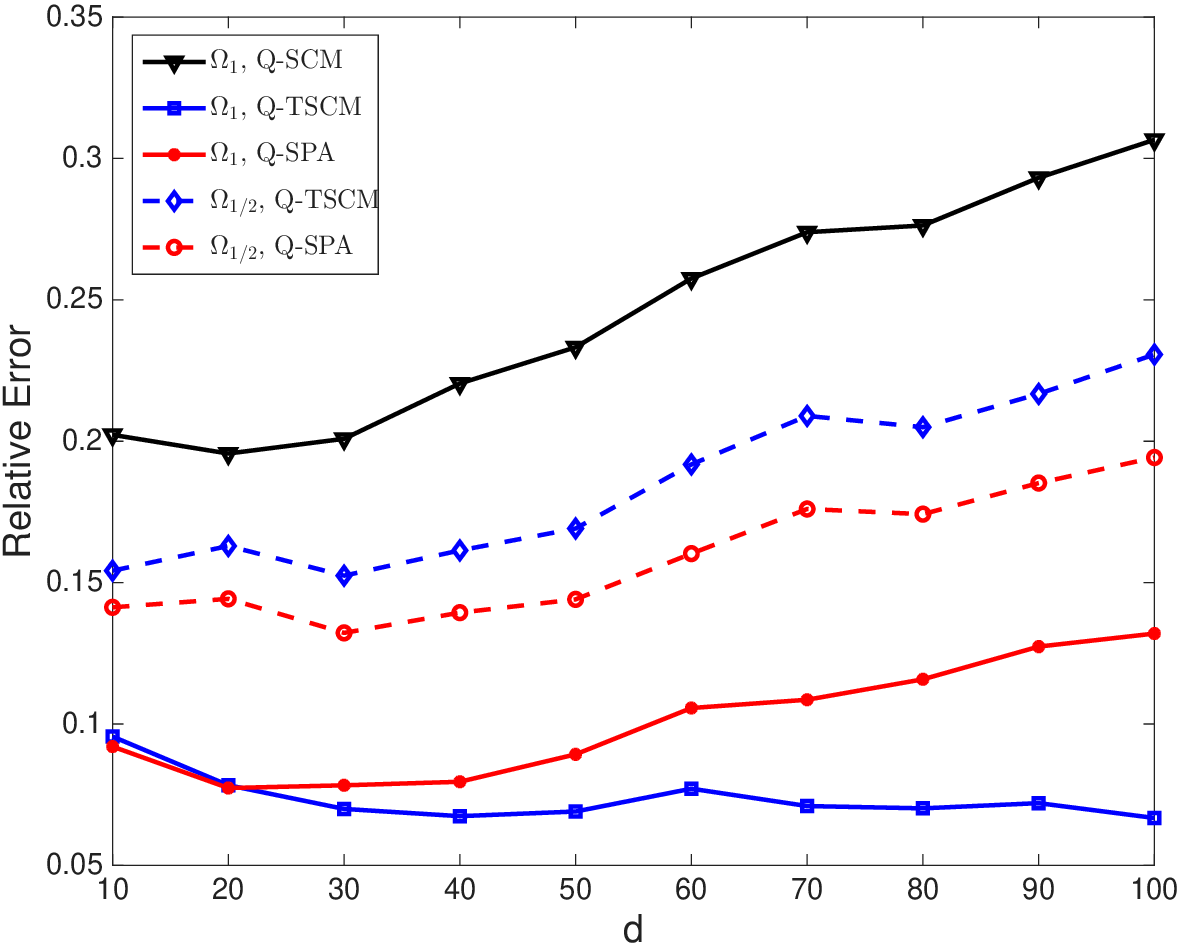}%
    \label{Fig_impact_d}}
  \caption{Impact of sample size $n$ and sample dimension $d$ on Toeplitz covariance estimation.}
  \label{Fig_impact_of_d_n}
\end{figure}

In \textit{Experiment 3}, we study the effects of the sample size $n$ and the dimension $d$.
We first fix $d=16$ and $\Delta = 5$, compute the relative error for varying $n$, and report the results on a log-log scale in Fig.~\ref{Fig_impact_n}\footnote{When $n$ is small, we introduce a diagonal perturbation $\epsilon \boldsymbol{I}_{|\Omega|}$ to the sample covariance matrix to compute its inverse before performing Q-SPA; after all, we know a priori that its expectation $\dot{\boldsymbol{R}} \succeq \frac{\|\boldsymbol{\Delta}\|_2^2}{4} \boldsymbol{I}_d$.}.
The curves are approximately linear with slope $-0.5$, consistent with the theoretical rate $\mathcal{O}(n^{-1/2})$.
For large $n$, Q-SPA typically achieves the best accuracy.
Notably, Q-SCM performs worse than Q-TSCM and can be inferior even to the sparse-observation case based on $\Omega_{1/2}$, highlighting the substantial gains achievable by exploiting the Toeplitz structure.
Next, fixing $n=500$, we plot the relative error as a function of $d$ in Fig.~\ref{Fig_impact_d}.
As $d$ increases, the errors of the Toeplitz-based estimators (Q-TSCM and Q-SPA) vary slowly; in particular, for fully observed Q-TSCM, the estimation error even decreases slightly.
This trend can be explained by the fact that, for a fixed lag $k$, a larger dimension provides more index pairs to estimate $\gamma_k$, thereby improving accuracy.
Moreover, in the fully observed scenario, when $d$ is large, Q-SPA performs less favourably than Q-TSCM.
One possible reason is that the condition number of the sample covariance matrix becomes large, whereas Q-SPA requires the use of the inverse of the sample covariance matrix, leading to reduced accuracy.
In more challenging partial-sampling regimes, however, Q-SPA yields markedly improved performance.

\begin{figure}[tbp]
  \centering
  \subfloat[Impact of $k$]{%
    \includegraphics[scale=0.28]{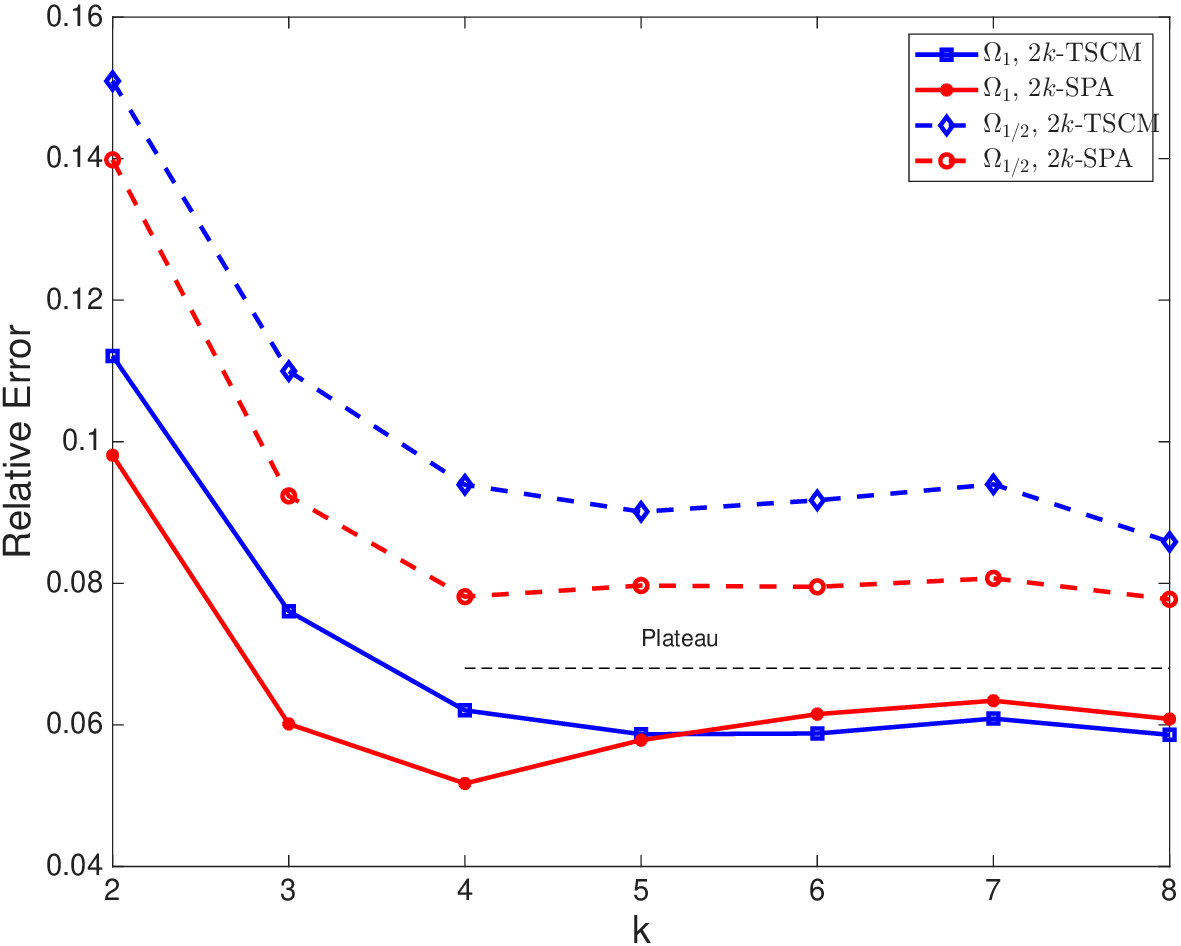}%
    \label{Fig_impact_k}}
  \hfill
  \subfloat[Comparison of 4-bit estimators]{%
    \includegraphics[scale=0.28]{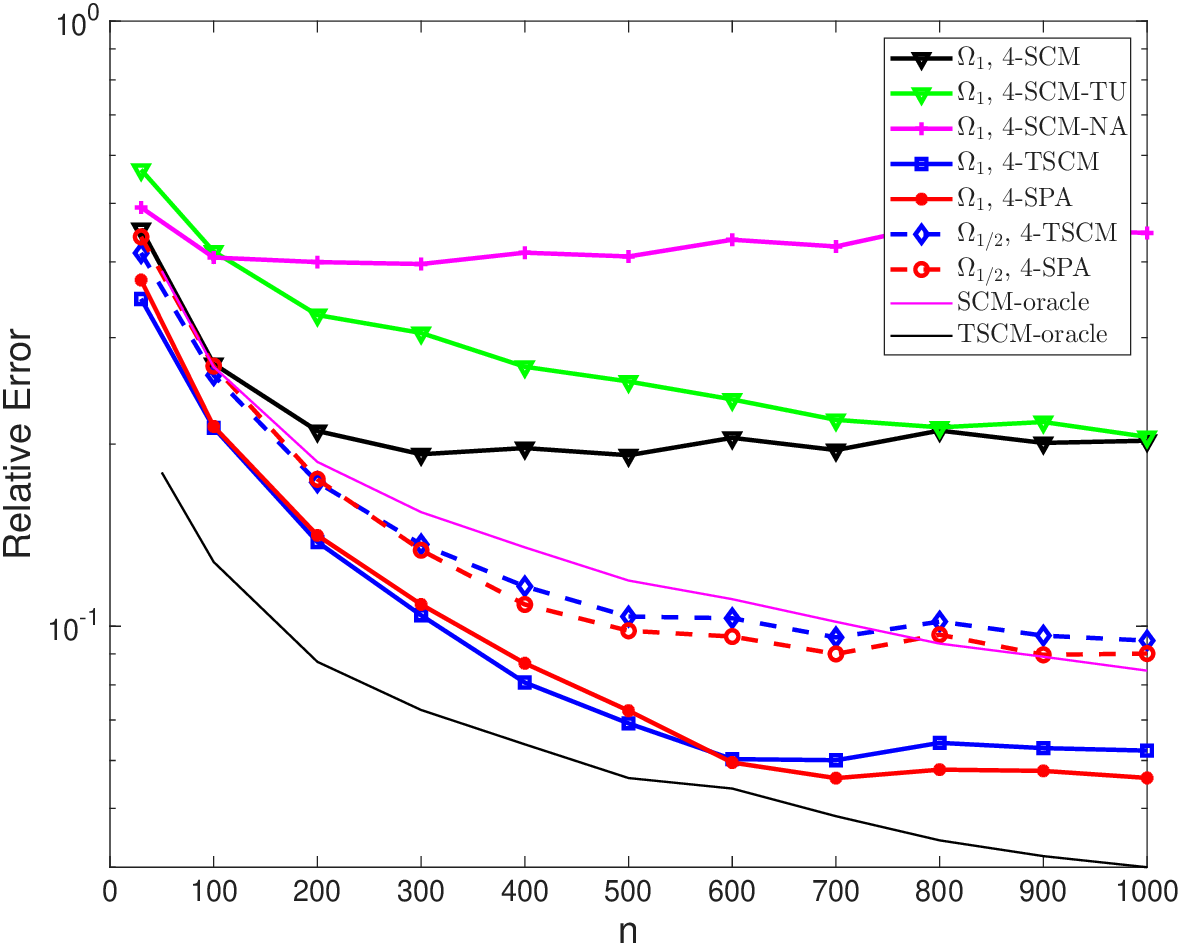}%
    \label{Fig_compare_4bit}}
  \caption{Impact of the number of bits $k$ on Toeplitz covariance estimation and comparison of 4-bit estimators.}
  \label{Fig_finite_bit}
\end{figure}

In \textit{Experiment 4}, we evaluate the finite-bit estimator $2k$-TSCM.
We set $d=16$, $n=500$, and plot the relative error of $2k$-TSCM and $2k$-SPA as a function of $k$ in Fig.~\ref{Fig_impact_k}.
As $k$ increases, the relative error first decreases and then levels off, forming a plateau.
Nevertheless, even in the challenging 4-bit case (i.e., $k=2$), the proposed estimators still achieve competitive accuracy.
In Fig.~\ref{Fig_compare_4bit}, we further compare the proposed methods with several baseline estimators in the 4-bit setting:
\begin{itemize}
    \item 4-SCM: the SCM under the triangular-dithered quantization framework proposed in \cite{chen2025parameter};
    \item 4-SCM-TU: the parameter-free variant introduced in \cite{chen2023quantizing, chen2025parameter};
    \item 4-SCM-NA: the SCM under the uniform-dithered quantization framework proposed in \cite{dirksen2024tuning, dirksen2025subspace, dirksen2022covariance}.
\end{itemize}
These competing estimators are applicable to fully observed data and do not exploit the Toeplitz structure.
Although they were originally formulated in the real-valued setting, they can be naturally extended to the complex-valued case since they do not rely on any structural constraints.
To quantify the performance loss caused by coarse quantization, we also include two full-precision oracle baselines under full observation, namely SCM-oracle and TSCM-oracle, which are computed from the corresponding unquantized samples.
As shown in Fig.~\ref{Fig_compare_4bit}, the proposed Toeplitz-based estimators (4-TSCM and 4-SPA) clearly outperform the existing few-bit SCM-type baselines. 
Moreover, they retain a clear advantage even under sparse observations $\Omega_{1/2}$. 
Compared with the full-precision oracle baselines, a performance gap due to coarse quantization remains. 
Nevertheless, the proposed Toeplitz-based few-bit estimators substantially reduce this gap, and can even outperform the unstructured full-precision SCM-oracle by exploiting the Toeplitz prior.

\subsection{Applications to DOA Estimation}

Consider a narrowband far-field array model with $K$ sources impinging on a $d$-sensor uniform linear array (ULA).
The received snapshot at time $t$ can be written as
\begin{equation}
    \boldsymbol{z}(t) = \boldsymbol{A}(\boldsymbol{f})\, \boldsymbol{s}(t) + \boldsymbol{n}(t),
\end{equation}
where $\boldsymbol{s}(t) \in \mathbb{C}^K$ denotes the source signal vector and $\boldsymbol{n}(t) \in \mathbb{C}^d$ is additive noise.
The steering matrix is $\boldsymbol{A}(\boldsymbol{f}) = \left[\boldsymbol{a}(f_1), \boldsymbol{a}(f_2), \ldots, \boldsymbol{a}(f_K)\right]$, where
\begin{equation}
    \boldsymbol{a}(f_k) = \big[1, e^{i 2\pi f_k}, \ldots, e^{i 2\pi (d-1) f_k}\big]^T .
\end{equation}
Note that the covariance matrix
\begin{equation}
    \boldsymbol{R} = \mathbb{E}\!\left[\boldsymbol{z}(t)\boldsymbol{z}(t)^H\right] = \boldsymbol{T} + \sigma_n^2 \boldsymbol{I}_d
\end{equation}
is Hermitian Toeplitz, where $\boldsymbol{T} = \sum_{k=1}^K \sigma_k^2 \boldsymbol{a}(f_k)\boldsymbol{a}(f_k)^H$ is the source covariance contribution, and $\sigma_k^2$ and $\sigma_n^2$ denote the power of the $k$-th source and the (spatially white) noise variance at each sensor, respectively.
We consider the case where only a sub-sample indexed by $\Omega$ is observed (corresponding to RSA), and we have access only to the triangularly quantized data $\dot{\boldsymbol{z}}_\Omega(t)$.
The goal of DOA estimation is to recover $\{f_k\}_{k=1}^K$, which has a one-to-one mapping to the directions, from the incomplete quantized observations $\dot{\boldsymbol{z}}_\Omega(t)$.
Given an estimate of the Toeplitz covariance matrix $\boldsymbol{T}$ (or $\boldsymbol{R}$), one may then apply subspace-based methods such as MUSIC or ESPRIT to obtain DOA estimates \cite{yang2018sparse}.

\begin{figure}[tbp]
  \centering
  \subfloat[MSE versus sample size $n$]{%
    \includegraphics[scale=0.28]{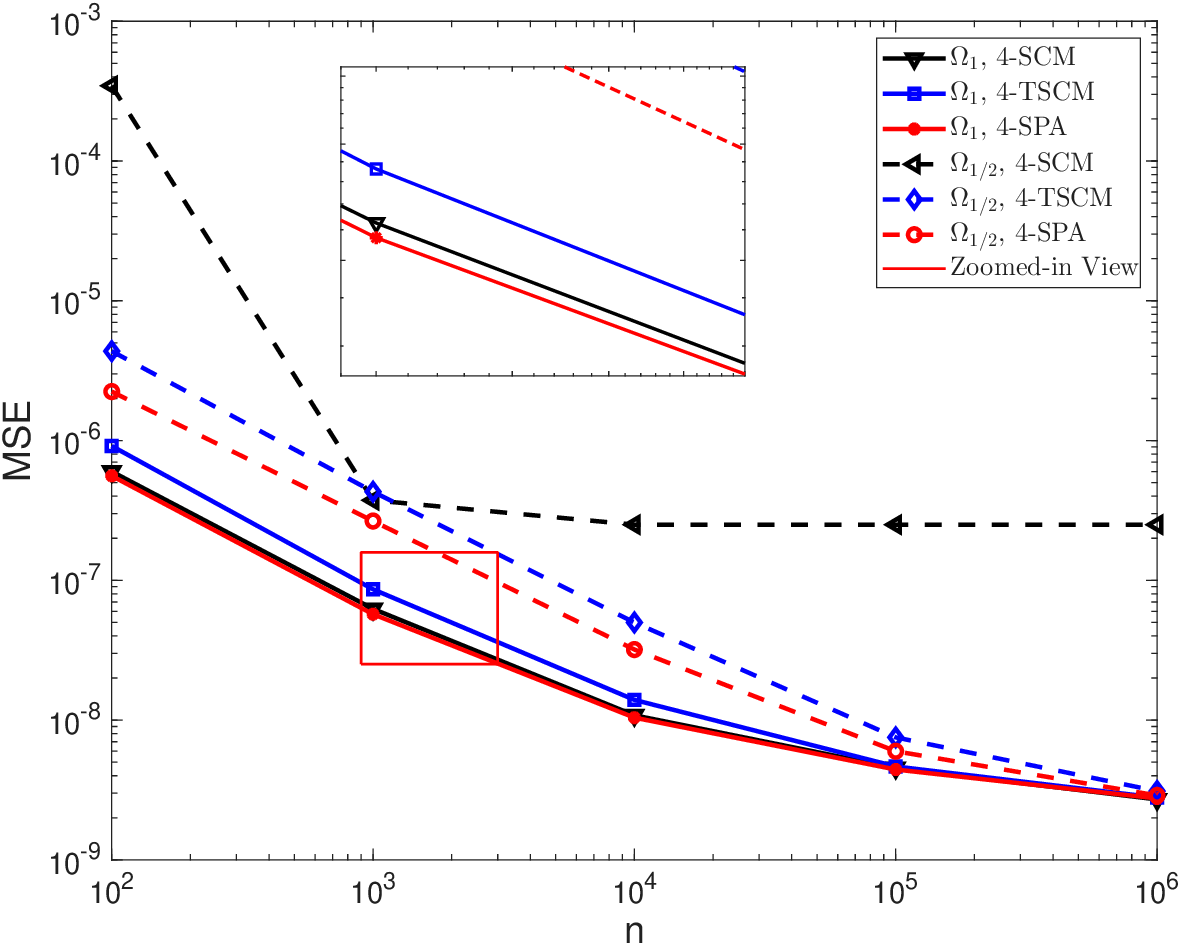}%
    \label{Fig_doa_mse_n}}
  \hfill
  \subfloat[MSE versus SNR]{%
    \includegraphics[scale=0.28]{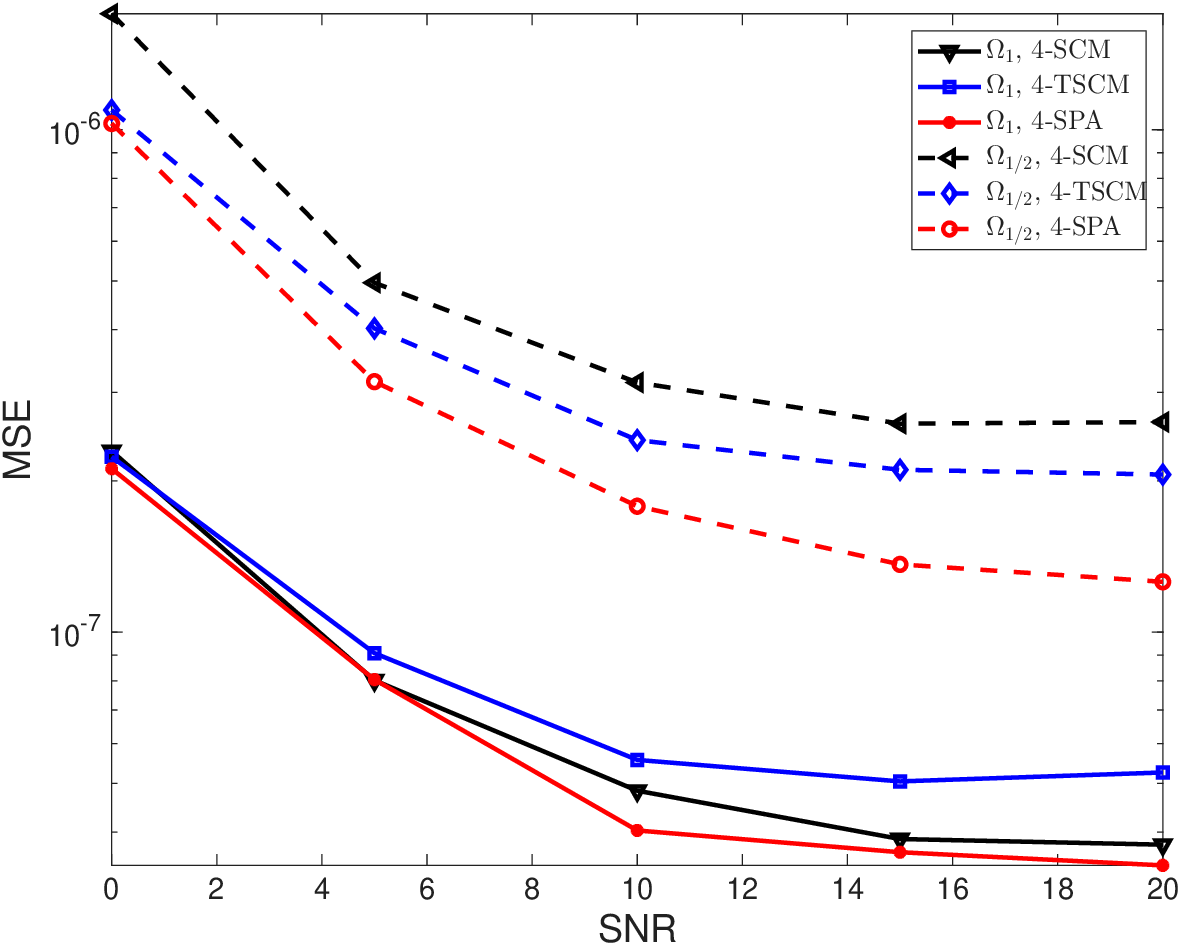}%
    \label{Fig_doa_mse_snr}}
  \caption{Application of 4-bit Toeplitz covariance estimation to DOA estimation.}
  \label{Fig_application_to_DOA}
\end{figure}

Building on the DOA model described above, we apply the proposed 4-bit quantized Toeplitz covariance estimation framework as a preprocessing step for DOA estimation.
Specifically, we first form a (Toeplitz-structured) covariance estimate from the 4-bit quantized observations and then apply MUSIC to estimate $\{f_k\}_{k=1}^K$.
In \textit{Experiment 5}, we set $K=5$, $\boldsymbol{f} = [0.08, 0.21, 0.37, 0.68, 0.81]^T$, and $d=16$.
Fig.~\ref{Fig_doa_mse_n} plots the mean squared error (MSE) of frequency estimation versus the number of snapshots (i.e., the sample size $n$) at a fixed signal-to-noise ratio (SNR).
As $n$ increases, the MSE of each method decreases monotonically.

Under full observations ($\Omega_1$), Q-SPA consistently achieves the lowest MSE, followed by Q-SCM\footnote{This is equivalent to the approach in \cite{dirksen2025subspace}, except that it is instantiated with ESPRIT. We refer readers to \cite{dirksen2025subspace} for further details.}, whereas Q-TSCM yields a comparatively higher MSE.
This suggests that, in the present setting, Toeplitz projection applied directly to the quantized sample covariance matrix does not necessarily improve spatial frequency estimation and may even lead to performance degradation.
For the sparse observation pattern ($\Omega_{1/2}$), Q-SPA still outperforms Q-TSCM, while Q-SCM performs worse than the other two.
As $n$ increases, the MSEs of Q-SPA and Q-TSCM under $\Omega_{1/2}$ approach their counterparts under full observations, indicating that DOA estimation accuracy close to the fully observed case can be achieved even with sparse observations.

As shown in Fig.~\ref{Fig_doa_mse_snr}, for all methods, the MSE decreases as the SNR increases.
Under both full and sparse observations, Q-SPA yields a lower MSE than Q-TSCM and Q-SCM; moreover, in the sparsely observed setting, Q-TSCM outperforms Q-SCM.
Overall, these results indicate that Q-SPA improves DOA estimation performance under low-bit quantization and can be readily applied to sparse-observation scenarios.
One possible explanation is that, in DOA estimation, the underlying Toeplitz covariance matrix is typically low-rank, which makes Q-SPA particularly effective.

\section{Conclusion}
This paper studies multi-bit Toeplitz covariance estimation under the triangular quantization framework.
We propose Q-TSCM and its finite-bit variant $2k$-TSCM to accommodate practical constraints such as missing data and coarse quantization, and analyze their non-asymptotic performance under the complex Gaussian assumption.
By establishing upper bounds on the covariance estimation error, we characterize how the quantization level and the chosen ruler affect the estimation accuracy.
In addition, we propose Q-SPA, which further improves empirical performance while enforcing the Toeplitz structure and positive semidefiniteness.
Our results suggest that triangular quantization provides a natural and effective approach for estimating Toeplitz covariance matrices.

Despite these advances, several challenges remain in quantized (Toeplitz) covariance estimation.
First, we only derive upper bounds in Theorem~\ref{Th_bound_T}; establishing a matching lower bound remains an open problem, even in the unquantized setting.
Second, it would be interesting to develop more refined quantization-aware estimation frameworks that use the statistical structure induced by dithering and coarse quantization to design more effective covariance estimation algorithms.
Third, it would be of interest to exploit the typical decay of Toeplitz covariance entries along the diagonals to design more effective quantized Toeplitz covariance estimators tailored to high-dimensional, low-sample regimes (or even single-snapshot observations).

\appendix

\subsection{Proof of Theorem~\ref{Th_bound_gamma}}\label{Pr_bound_gamma}
For any $s = 0, 1, \ldots, d-1$ and $(j, k) \in \Omega_s$, we have
\begin{equation}
    \begin{aligned}
        &\dot{z}_j^{(l)} \dot{z}_k^{(l)^*} - \dfrac{\|\boldsymbol{\Delta}\|_2^2}{4} \delta_s - \gamma_s\\
        =& \dot{x}_j^{(l)} \dot{x}_k^{(l)} - \mathbb{E}\left[\dot{x}_j^{(l)} \dot{x}_k^{(l)}\right] + \dot{y}_j^{(l)} \dot{y}_k^{(l)} - \mathbb{E}\left[\dot{y}_j^{(l)} \dot{y}_k^{(l)}\right]\\
        &\quad + i\left(\dot{y}_j^{(l)} \dot{x}_k^{(l)} - \mathbb{E}\left[\dot{y}_j^{(l)} \dot{x}_k^{(l)}\right] - \dot{x}_j^{(l)} \dot{y}_k^{(l)} + \mathbb{E}\left[\dot{x}_j^{(l)} \dot{y}_k^{(l)}\right]\right) \\
        \triangleq & E_{x_j, x_k}^{(l)} + E_{y_j, y_k}^{(l)} + i \left(E_{y_j, x_k}^{(l)} - E_{x_j, y_k}^{(l)}\right),
    \end{aligned}
\end{equation}
and hence
\begin{equation}
    \begin{aligned}
        &\widehat{\gamma}_s - \gamma_s\\
        =& \dfrac{1}{n |\Omega_s|} \sum_{l=1}^{n} \sum_{(j, k) \in \Omega_s} \left( \dot{z}_j^{(l)} \dot{z}_k^{(l)^*} - \dfrac{\|\boldsymbol{\Delta}\|_2^2}{4} \delta_s - \gamma_s\right) \\
        \triangleq & E_{x, x} + E_{y, y} + i \left(E_{y, x} - E_{x, y}\right),
    \end{aligned}
\end{equation}
where, for any $\{\star, \cdot\} \in \{x,y\}\times\{x,y\}$, we define
\begin{equation*}
    \left\{
    \begin{aligned}
        E_{\star, \cdot}^{(l)} &= \dfrac{1}{|\Omega_s|}\sum_{(j, k) \in \Omega_s} E_{\star_j, \cdot_k}^{(l)},\\
        E_{\star, \cdot} &= \dfrac{1}{n} \sum_{l = 1}^n E_{\star, \cdot}^{(l)}.
    \end{aligned}
    \right.
\end{equation*}
Then, for any $t > 0$, we have
\begin{equation}
    \label{Eq_Bound_For_Gamma_s_Parts}
    \begin{aligned}
        &\mathbb{P}\left( |\widehat{\gamma}_s - \gamma_s| > t \right)\\
        =& \mathbb{P}\left( \left(E_{x, x} + E_{y, y}\right)^2 + \left(E_{y, x} - E_{x, y}\right)^2 > t^2 \right)\\
        \leq & \mathbb{P}\left( \left|E_{x, x} + E_{y, y}\right| > \dfrac{t}{\sqrt{2}} \right) + \mathbb{P}\left( \left|E_{y, x} - E_{x, y}\right| > \dfrac{t}{\sqrt{2}} \right)\\
        \leq & \sum_{\{\star, \cdot\} = \{x, y\}}\mathbb{P}\left( \left|E_{\star, \cdot} \right| > \dfrac{t}{2\sqrt{2}} \right).
    \end{aligned}
\end{equation}
Moreover, we note that
\begin{equation}
    \begin{aligned}
        \| \dot{x}_j^{(l)} \|_{\psi_2} =& \| x_j^{(l)} + \omega_{r, j}^{(l)} + \tau_{r, j}^{(l)} \|_{\psi_2} \\
        \leq& \| x_j^{(l)} \|_{\psi_2} + \| \omega_{r, j}^{(l)} \|_{\psi_2} + \| \tau_{r, j}^{(l)} \|_{\psi_2}\\
        \leq& \sqrt{|\gamma_0|} + 2\Delta_r\\
        \leq& \|\boldsymbol{T}\|_2^{1/2} + 2 \| \boldsymbol{\Delta}\|_2\\
        \triangleq& \mathcal{K}_{\boldsymbol{\Delta}, \boldsymbol{T}},
    \end{aligned}
\end{equation}
and similarly, $\| \dot{y}_j^{(l)} \|_{\psi_2} \leq \mathcal{K}_{\boldsymbol{\Delta}, \boldsymbol{T}}$, implying that both $\dot{x}_j^{(l)}$ and $\dot{y}_j^{(l)}$ are sub-Gaussian random variables.
Consequently, the product $\dot{x}_j^{(l)} \dot{x}_k^{(l)}$ is sub-exponential and satisfies
\begin{equation}
    \begin{aligned}
        \| \dot{x}_j^{(l)} \dot{x}_k^{(l)} \|_{\psi_1} &\leq \| \dot{x}_j^{(l)} \|_{\psi_2} \| \dot{x}_k^{(l)} \|_{\psi_2} \\
        &\leq \left(\|\boldsymbol{T}\|_2^{1/2} + 2 \| \boldsymbol{\Delta}\|_2\right)^2 = \mathcal{K}_{\boldsymbol{\Delta}, \boldsymbol{T}}^2.
    \end{aligned}
\end{equation}
Therefore, there exists a universal constant $C_{\psi} \geq 1$ such that \cite{vershynin2018high}
\begin{equation}
    \begin{aligned}
        \left\| E_{x, x}^{(l)} \right\|_{\psi_1} &= \left\| \dfrac{1}{|\Omega_s|} \sum_{(j, k) \in \Omega_s} E_{x_j, x_k}^{(l)}\right\|_{\psi_1}\\
        &\leq \dfrac{1}{|\Omega_s|} \sum_{(j, k) \in \Omega_s} \left\| \dot{x}_j^{(l)} \dot{x}_k^{(l)} - \mathbb{E}\left[\dot{x}_j^{(l)} \dot{x}_k^{(l)}\right]\right\|_{\psi_1}\\
        &\leq C_{\psi} \mathcal{K}_{\boldsymbol{\Delta}, \boldsymbol{T}}^2.
    \end{aligned}
\end{equation}
By Bernstein's inequality, there exists a universal constant $c_{\mathrm{Ber}} > 0$ such that, for any $t > 0$,
\begin{equation}
    \begin{aligned}
        &\mathbb{P}\left( \left|E_{x, x}\right| > \dfrac{t}{2\sqrt{2}} \right) = \mathbb{P}\left( \left|\sum_{l=1}^n E_{x, x}^{(l)}\right| > \dfrac{nt}{2\sqrt{2}} \right)\\
        \leq& 2 \exp{\left( -c_{\mathrm{Ber}} \min \left\{ \dfrac{n^2 t^2}{8 \sum_{l=1}^n \| E_{x, x}^{(l)} \|_{\psi_1}^2}, \dfrac{nt}{2\sqrt{2} \max_l \| E_{x, x}^{(l)} \|_{\psi_1}} \right\} \right)}\\
        \leq& 2 \exp{\left( -c_{\mathrm{Ber}} n \min \left\{ \dfrac{t^2}{8C_{\psi}^2 \mathcal{K}_{\boldsymbol{\Delta}, \boldsymbol{T}}^4}, \dfrac{t}{2\sqrt{2}C_{\psi} \mathcal{K}_{\boldsymbol{\Delta}, \boldsymbol{T}}^2} \right\} \right)}.
    \end{aligned}
\end{equation}
Similarly, the same bound applies to $\mathbb{P}\bigl( \lvert E_{\star,\cdot} \rvert > \tfrac{t}{2\sqrt{2}} \bigr)$ for all $(\star,\cdot) \in \{x,y\}\times\{x,y\}$.
Combining these bounds with \eqref{Eq_Bound_For_Gamma_s_Parts}, we obtain \eqref{Eq_Bound_For_Gamma} only by noting that
\begin{equation*}
    \min \left\{ \dfrac{t^2}{8C_{\psi}^2 \mathcal{K}_{\boldsymbol{\Delta}, \boldsymbol{T}}^4}, \dfrac{t}{2\sqrt{2}C_{\psi} \mathcal{K}_{\boldsymbol{\Delta}, \boldsymbol{T}}^2} \right\} \geq \min \left\{ \dfrac{t^2}{8C_{\psi}^2 \mathcal{K}_{\boldsymbol{\Delta}, \boldsymbol{T}}^4}, 1 \right\}
\end{equation*}
and taking $c = c_{\mathrm{lag}} = c_{\mathrm{Ber}}$, $C = C_{\mathrm{lag}} = 8C_{\psi}^2$.
Thus, the universal constants $c$ and $C$ in Theorem~\ref{Th_bound_gamma} can be taken as $c=c_{\mathrm{lag}}$ and $C=C_{\mathrm{lag}}$, respectively.
Substituting $t \asymp \mathcal{K}_{\boldsymbol{\Delta}, \boldsymbol{T}}^2\sqrt{\delta / n}$ into \eqref{Eq_Bound_For_Gamma} yields \eqref{Eq_error_gamma}, which completes the proof.

\subsection{Proof of Theorem~\ref{Th_bound_T}}\label{Pr_bound_T}
To analyze the non-asymptotic performance of $\widehat{\boldsymbol{T}}$, we begin with the following lemma\cite{grenander1958toeplitz}.

\begin{lemma}
    \label{La_spectral_density}
    Let $\boldsymbol{e} = (e_{-d+1}, \ldots, e_0, \ldots, e_{d-1})^T \in \mathbb{C}^{2d-1}$ and $e_{-s} = e_s^*$, then the spectral norm of the Toeplitz matrix $\mathcal{T}(\boldsymbol{e})$ is bounded by the supremum norm of its associated spectral density $L_{\boldsymbol{e}} (\theta)$, i.e.,
    \begin{equation}
        \| \mathcal{T}(\boldsymbol{e}) \|_2 \leq \sup_{\theta \in [0,1]} \left|L_{\boldsymbol{e}} (\theta)\right|,
    \end{equation}
    where the spectral density is defined by
    \begin{equation}
        L_{\boldsymbol{e}} (\theta) = \sum_{s = 1 - d}^{d - 1} e_s e^{i 2 \pi s \theta} = e_0 + \sum_{s = 1}^{d - 1} \left(e_s e^{i 2 \pi s \theta} + e_s^* e^{- i 2 \pi s \theta}\right).
    \end{equation}
\end{lemma}

Before proceeding to the proof of Theorem~\ref{Th_bound_T}, we first establish an auxiliary concentration lemma for quadratic forms of the lifted quantized observation vector. 
The lemma provides a substitute for a direct application of the classical Hanson--Wright inequality, which is not justified here since the coordinates of $\dot{\boldsymbol u}_\Omega^{(l)}$ are generally dependent.

\begin{lemma}
    \label{Lem_W_H_type_inequality}
    Under the setting of Theorem~\ref{Th_bound_T}, let
    \begin{equation*}
        \mathcal{K}_{\boldsymbol{\Delta}, \boldsymbol{T}} = \|\boldsymbol T\|_2^{1/2} + 2\|\boldsymbol\Delta\|_2.
    \end{equation*}
    For any nonzero deterministic real symmetric matrix $\boldsymbol{\Lambda} \in \mathbb{R}^{2|\Omega| \times 2|\Omega|}$, define
    \begin{equation}
        \boldsymbol{u}_\Omega^{(l)} = 
        \begin{pmatrix}
            \Re(\boldsymbol{z}_\Omega^{(l)})\\
            \Im(\boldsymbol{z}_\Omega^{(l)})
        \end{pmatrix}
        \qquad
        \dot{\boldsymbol{u}}_\Omega^{(l)} = 
        \begin{pmatrix}
            \Re(\dot{\boldsymbol{z}}_\Omega^{(l)})\\
            \Im(\dot{\boldsymbol{z}}_\Omega^{(l)})
        \end{pmatrix}.
    \end{equation}
    Then there exist universal constants $c_{\mathrm{qhw}} ,C_{\mathrm{qhw}}>0$ such that, for all $t>0$,
    \begin{equation}
        \begin{aligned}
            \mathbb{P}\left(\left|\frac{1}{n} \sum_{l=1}^{n} \dot{\boldsymbol{u}}_\Omega^{(l)^T} \boldsymbol{\Lambda} \dot{\boldsymbol{u}}_\Omega^{(l)} - \mathbb{E}\left[\dot{\boldsymbol{u}}_\Omega^{(l)^T} \boldsymbol{\Lambda} \dot{\boldsymbol{u}}_\Omega^{(l)}\right]\right| > t\right) \\
            \leq C_{\mathrm{qhw}} \exp \left(-c_{\mathrm{qhw}} n \min\left\{\frac{t^2}{\mathcal{K}_{\boldsymbol{\Delta}, \boldsymbol{T}}^4 \|\boldsymbol{\Lambda}\|_F^2}, 1\right\}\right).
        \end{aligned}
    \end{equation}
\end{lemma}
\begin{proof}
    Write $\boldsymbol{\xi}_\Omega^{(l)} = \dot{\boldsymbol{u}}_\Omega^{(l)} - \boldsymbol{u}_\Omega^{(l)}$.
    Since $\boldsymbol{\Lambda}$ is symmetric,
    \begin{equation}
        \dot{\boldsymbol u}_\Omega^{(l)^T} \boldsymbol{\Lambda} \dot{\boldsymbol u}_\Omega^{(l)}
        = \boldsymbol u_\Omega^{(l)^T} \boldsymbol{\Lambda} \boldsymbol u_\Omega^{(l)}
        + 2 \boldsymbol u_\Omega^{(l)^T} \boldsymbol{\Lambda} \boldsymbol{\xi}_\Omega^{(l)}
        + \boldsymbol{\xi}_\Omega^{(l)^T} \boldsymbol{\Lambda} \boldsymbol{\xi}_\Omega^{(l)} . 
    \end{equation}
    Moreover, the property of triangular quantization gives that $\mathbb{E}[\boldsymbol{\xi}_\Omega^{(l)}\mid \boldsymbol u_\Omega^{(l)}] = \boldsymbol{0}$.
    Therefore,
    \begin{equation}
        \mathbb{E}\left[\boldsymbol{u}_\Omega^{(l)^T} \boldsymbol{\Lambda} \boldsymbol{\xi}_\Omega^{(l)}\right]
        = \mathbb{E}\left[\boldsymbol{u}_\Omega^{(l)^T} \boldsymbol{\Lambda} \mathbb{E}\left[ \boldsymbol{\xi}_\Omega^{(l)} | \boldsymbol{u}_\Omega^{(l)}\right]\right]
        = 0.
    \end{equation}
    Consequently,
    \begin{equation}
        \dot{\boldsymbol{u}}_\Omega^{(l)^T} \boldsymbol{\Lambda} \dot{\boldsymbol{u}}_\Omega^{(l)} - \mathbb{E}\left[\dot{\boldsymbol{u}}_\Omega^{(l)^T} \boldsymbol{\Lambda} \dot{\boldsymbol{u}}_\Omega^{(l)}\right]
        = G^{(l)} + H^{(l)} + R^{(l)},
    \end{equation}
    where
    \begin{equation}
        \begin{aligned}
            &G^{(l)} = \boldsymbol{u}_\Omega^{(l)^T} \boldsymbol{\Lambda} \boldsymbol{u}_\Omega^{(l)} - \mathbb{E}\left[\boldsymbol{u}_\Omega^{(l)^T} \boldsymbol{\Lambda} \boldsymbol{u}_\Omega^{(l)}\right], \\
            &H^{(l)} = 2 \boldsymbol{u}_\Omega^{(l)^T} \boldsymbol{\Lambda} \boldsymbol{\xi}_\Omega^{(l)}, \\
            &R^{(l)} = \boldsymbol{\xi}_\Omega^{(l)^T} \boldsymbol{\Lambda} \boldsymbol{\xi}_\Omega^{(l)} - \mathbb{E}\left[\boldsymbol{\xi}_\Omega^{(l)^T} \boldsymbol{\Lambda} \boldsymbol{\xi}_\Omega^{(l)}\right].
        \end{aligned}
    \end{equation}
    We control the three terms separately.

    \paragraph{Control of $G^{(l)}$}
    Let $\boldsymbol{\Sigma}_\Omega = \operatorname{Cov}\left(\boldsymbol u_\Omega^{(l)}\right)$.
    Since $\boldsymbol u_\Omega^{(l)}$ is the real lift of $\boldsymbol z_\Omega^{(l)}$, we have $\|\boldsymbol{\Sigma}_\Omega\|_2 \leq \|\boldsymbol{T}_\Omega\|_2 \leq \|\boldsymbol{T}\|_2$.
    Write
    \begin{equation}
        \boldsymbol{u}_\Omega^{(l)} = \boldsymbol{\Sigma}_\Omega^{1/2} \boldsymbol{g}^{(l)},
        \qquad \boldsymbol{g}^{(l)} \sim \mathcal{N}(0, \boldsymbol{I}_{2|\Omega|}),
    \end{equation}
    and define
    \begin{equation}
        \boldsymbol{\Lambda}' = \boldsymbol{\Sigma}_\Omega^{1/2} \boldsymbol{\Lambda} \boldsymbol{\Sigma}_\Omega^{1/2}.
    \end{equation}
    Then
    \begin{equation}
        G^{(l)} = \boldsymbol{g}^{(l)^T} \boldsymbol{\Lambda}' \boldsymbol{g}^{(l)} - \mathbb{E}[\boldsymbol{g}^{(l)^T} \boldsymbol{\Lambda}' \boldsymbol{g}^{(l)}].
    \end{equation}
    Applying the Hanson--Wright inequality \cite{rudelson2013hanson}, we obtain
    \begin{equation}
        \begin{aligned}
            &\mathbb{P}\left(\left|\frac{1}{n} \sum_{l=1}^{n} G^{(l)}\right| > t\right) \\
            \leq& 2 \exp \left(-c_{\mathrm{hw}} n \min\left\{\frac{t^2}{\|\boldsymbol{\Lambda}'\|_F^2}, \frac{t}{\|\boldsymbol{\Lambda}'\|_2}\right\}\right).
        \end{aligned}
    \end{equation}
    Moreover,
    \begin{equation}
        \begin{aligned}
            \|\boldsymbol{\Lambda}'\|_F 
            =& \left\|\boldsymbol{\Sigma}_\Omega^{1/2} \boldsymbol{\Lambda} \boldsymbol{\Sigma}_\Omega^{1/2} \right\|_F
            \leq \|\boldsymbol{\Sigma}_\Omega\|_2 \|\boldsymbol{\Lambda}\|_F \\
            &\leq \|\boldsymbol{T}\|_2 \|\boldsymbol{\Lambda}\|_F
            \leq K^2 \|\boldsymbol{\Lambda}\|_F,
        \end{aligned}
    \end{equation}
    and $\|\boldsymbol{\Lambda}'\|_2 \leq \|\boldsymbol{\Lambda}'\|_F$.
    Therefore,
    \begin{equation}
        \begin{aligned}
            \mathbb{P}&\left(\left|\frac{1}{n} \sum_{l=1}^{n} G^{(l)}\right| > t\right) \\
            &\leq 2 \exp \left(-c_{\mathrm{hw}} n \min\left\{\frac{t^2}{K^4 \|\boldsymbol{\Lambda}\|_F^2}, 1\right\}\right).
        \end{aligned}
    \end{equation}

    \paragraph{Control of $H^{(l)}$}
    Let $\mathcal{S} = \sigma(\boldsymbol{u}_\Omega^{(1)}, \boldsymbol{u}_\Omega^{(2)}, \ldots, \boldsymbol{u}_\Omega^{(n)})$ collect the information carried by $\boldsymbol{u}_\Omega^{(l)}$, $l=1,2,\ldots,n$.
    For any $j=1,\ldots,2|\Omega|$, define
    \begin{equation}
        \Delta_j = \left\{
            \begin{aligned}
                &\Delta_r, \quad 1 \leq j \leq |\Omega|,\\
                &\Delta_i, \quad |\Omega| < j \leq 2 |\Omega|,
            \end{aligned}
        \right.
    \end{equation}
    and let
    \begin{equation}
        \boldsymbol{D}_{\boldsymbol{\Delta}} = \operatorname{diag}(\Delta_1, \Delta_2, \ldots, \Delta_{2|\Omega|}).
    \end{equation}
    Conditionally on $\mathcal{S}$, the entries of $\boldsymbol{\xi}^{(1)}, \boldsymbol{\xi}^{(2)}, \ldots, \boldsymbol{\xi}^{(n)}$ are independent, centered, and sub-Gaussian.
    More precisely, for all $l, j$,
    \begin{equation}
        \|\xi^{(l)}_{\Omega, j} \mid \mathcal{S}\|_{\psi_2} 
        = \|\omega^{(l)}_{\Omega, j} + \tau^{(l)}_{\Omega, j} \mid \mathcal{S}\|_{\psi_2}
        \leq 2 \Delta_j,
    \end{equation}
    and $\mathbb{E}[\xi^{(l)}_{\Omega, j} | \mathcal{S}] = 0$.

    Since $\boldsymbol{\Lambda}$ is symmetric, conditionally on $\mathcal{S}$,
    \begin{equation}
        \boldsymbol{u}_\Omega^{(l)^T} \boldsymbol{\Lambda} \boldsymbol{\xi}_\Omega^{(l)}
        = \left(\boldsymbol{\Lambda} \boldsymbol{u}_\Omega^{(l)}\right)^T \boldsymbol{\xi}_\Omega^{(l)}
        \triangleq \sum_{j=1}^{2|\Omega|} a_{l,j} \xi_{\Omega, j}^{(l)},
    \end{equation}
    where $a_{l,j} = (\boldsymbol{\Lambda} \boldsymbol{u}_\Omega^{(l)})_j$ and $\xi_{\Omega, j}^{(l)} = (\boldsymbol{\xi}_\Omega^{(l)})_j$.
    Given $\mathcal{S}$, the coefficients $a_{l,j}$ are deterministic.
    Therefore, by the standard bound for weighted sums of independent sub-Gaussian random variables \cite{vershynin2018high}, it follows that
    \begin{equation}
        \begin{aligned}
            &\left\|\boldsymbol{u}_\Omega^{(l)^T} \boldsymbol{\Lambda} \boldsymbol{\xi}_\Omega^{(l)} \right\|_{\psi_2}^2
            \leq C_{\mathrm{sum}} \sum_{j=1}^{2|\Omega|} \left\|a_{l,j} \xi_{\Omega, j}^{(l)} \right\|_{\psi_2}^2 \\
            =& C_{\mathrm{sum}} \sum_{j=1}^{2|\Omega|} |a_{l,j}|^2 \left\|\xi_{\Omega, j}^{(l)} \right\|_{\psi_2}^2 
            \leq 4C_{\mathrm{sum}} \sum_{j=1}^{2|\Omega|} |a_{l,j}|^2 \Delta_j^2 \\
            =& 4C_{\mathrm{sum}} \left\|\boldsymbol{D}_{\boldsymbol{\Delta}} \boldsymbol{\Lambda} \boldsymbol{u}_\Omega^{(l)} \right\|_2^2 
            \leq 4C_{\mathrm{sum}} \left\|\boldsymbol{D}_{\boldsymbol{\Delta}}\right\|_2^2 \left\|\boldsymbol{\Lambda} \boldsymbol{u}_\Omega^{(l)} \right\|_2^2.
        \end{aligned}
    \end{equation}
    Hoeffding's inequality \cite{vershynin2018high} then gives
    \begin{equation}
        \label{Eq_lemma_h_bound_i}
        \begin{aligned}
            &\mathbb{P}\left(\left|\frac{1}{n}\sum_{l=1}^{n} H^{(l)}\right| > t \middle| \mathcal{S}\right) \\
            =& \mathbb{P}\left(\left|\frac{1}{n}\sum_{l=1}^{n} \boldsymbol{u}_\Omega^{(l)^T} \boldsymbol{\Lambda} \boldsymbol{\xi}_\Omega^{(l)}\right| > \frac{t}{2} \middle| \mathcal{S}\right) \\
            \leq& 2 \exp\left(\frac{-c_{\mathrm{Hoe}}n^2t^2}{4\sum_{l=1}^{n}\|\boldsymbol{u}_\Omega^{(l)^T} \boldsymbol{\Lambda} \boldsymbol{\xi}_\Omega^{(l)}\|_{\psi_2}^2}\right) \\
            \leq& 2 \exp\left(\frac{-c_{H,1}n^2t^2}{\left\|\boldsymbol{D}_{\boldsymbol{\Delta}}\right\|_2^2 \sum_{l=1}^{n}\|\boldsymbol{\Lambda} \boldsymbol{u}_\Omega^{(l)} \|_{2}^2}\right),
        \end{aligned}
    \end{equation}
    where $c_{H,1} = c_{\mathrm{Hoe}} / 16C_{\mathrm{sum}}$.

    It remains to control $\sum_{l=1}^{n}\|\boldsymbol{\Lambda} \boldsymbol{u}_\Omega^{(l)} \|_{2}^2$.
    Let
    \begin{equation}
        \boldsymbol{\Lambda}'' = \boldsymbol{\Sigma}_\Omega^{1/2} \boldsymbol{\Lambda}^2 \boldsymbol{\Sigma}_\Omega^{1/2}
        \succeq \boldsymbol{0}.
    \end{equation}
    Then
    \begin{equation}
        \left\|\boldsymbol{\Lambda} \boldsymbol{u}_\Omega^{(l)} \right\|_{2}^2
        = \boldsymbol{g}^{(l)^T} \boldsymbol{\Lambda}'' \boldsymbol{g}^{(l)},
    \end{equation}
    and $\mathbb{E}\left[\boldsymbol{g}^{(l)^T} \boldsymbol{\Lambda}'' \boldsymbol{g}^{(l)}\right] = \mathrm{tr}\left(\boldsymbol{\Lambda}''\right)$.
    Applying the Hanson--Wright inequality \cite{rudelson2013hanson}, we obtain
    \begin{equation}
        \label{Eq_lamme_H_u_bound}
        \begin{aligned}
            &\mathbb{P}\left(\left|\frac{1}{n} \sum_{l=1}^{n} \left\|\boldsymbol{\Lambda} \boldsymbol{u}_\Omega^{(l)} \right\|_{2}^2 - \mathrm{tr}\left(\boldsymbol{\Lambda}''\right)\right| > s\right) \\
            =& \mathbb{P}\left(\left|\frac{1}{n} \sum_{l=1}^{n} \boldsymbol{g}^{(l)^T} \boldsymbol{\Lambda}'' \boldsymbol{g}^{(l)} - \mathbb{E}\left[\boldsymbol{g}^{(l)^T} \boldsymbol{\Lambda}'' \boldsymbol{g}^{(l)}\right]\right| > s\right) \\
            \leq& 2 \exp \left(-c_{\mathrm{hw}} n \min\left\{\frac{s^2}{\|\boldsymbol{\Lambda}''\|_F^2}, \frac{s}{\|\boldsymbol{\Lambda}''\|_2}\right\}\right) \\
            \leq& 2 \exp \left(-c_{\mathrm{hw}} n \min\left\{\frac{s^2}{\|\boldsymbol{\Lambda}''\|_F^2}, 1\right\}\right).
        \end{aligned}
    \end{equation}
    by noting that $\|\boldsymbol{\Lambda}''\|_2 \leq \|\boldsymbol{\Lambda}''\|_F$.
    Moreover,
    \begin{equation}
        \mathrm{tr}\left(\boldsymbol{\Lambda}''\right) 
        = \mathrm{tr}\left(\boldsymbol{\Lambda}^2 \boldsymbol{\Sigma}_\Omega\right) 
        \leq \|\boldsymbol{\Sigma}_\Omega\|_2 \|\boldsymbol{\Lambda}\|_F^2
        \leq \|\boldsymbol{T}\|_2 \|\boldsymbol{\Lambda}\|_F^2,
    \end{equation}
    and
    \begin{equation}
        \begin{aligned}
            \|\boldsymbol{\Lambda}''\|_F 
            \leq \|\boldsymbol{\Sigma}_\Omega\|_2 \|\boldsymbol{\Lambda}^2\|_F 
            \leq \|\boldsymbol{T}\|_2 \|\boldsymbol{\Lambda}\|_2 \|\boldsymbol{\Lambda}\|_F
            \leq \|\boldsymbol{T}\|_2 \|\boldsymbol{\Lambda}\|_F^2.
        \end{aligned}
    \end{equation}
    By choosing $s = \|\boldsymbol{T}\|_2 \|\boldsymbol{\Lambda}\|_F^2$ in \eqref{Eq_lamme_H_u_bound}, we obtain
    \begin{equation}
        \label{Eq_lemma_h_bound_ii}
        \mathbb{P}\left(\frac{1}{n} \sum_{l=1}^{n} \left\|\boldsymbol{\Lambda} \boldsymbol{u}_\Omega^{(l)} \right\|_{2}^2 > 2 \|\boldsymbol{T}\|_2 \|\boldsymbol{\Lambda}\|_F^2\right) \leq 2 e^{-c_{\mathrm{hw}}n}.
    \end{equation}
    On the complementary event,
    \begin{equation}
        \sum_{l=1}^{n} \left\|\boldsymbol{\Lambda} \boldsymbol{u}_\Omega^{(l)} \right\|_{2}^2 
        \leq 2 n \|\boldsymbol{T}\|_2 \|\boldsymbol{\Lambda}\|_F^2.
    \end{equation}
    Thus
    \begin{equation}
        \begin{aligned}
            &\mathbb{P}\left(\left|\frac{1}{n}\sum_{l=1}^{n} H^{(l)}\right| > t\right) \\
            \leq& 2 e^{-c_{\mathrm{hw}}n} + 2 \exp\left(\frac{-c_{H,1}nt^2}{\left\|\boldsymbol{D}_{\boldsymbol{\Delta}}\right\|_2^2 \|\boldsymbol{T}\|_2 \|\boldsymbol{\Lambda}\|_F^2}\right).
        \end{aligned}    
    \end{equation}
    Since
    \begin{equation}
        \left\|\boldsymbol{D}_{\boldsymbol{\Delta}}\right\|_2 \leq \|\boldsymbol{\Delta}\|_2,
    \end{equation}
    and therefore
    \begin{equation}
        \left\|\boldsymbol{D}_{\boldsymbol{\Delta}}\right\|_2^2 \|\boldsymbol{T}\|_2 
        \leq \left(\|\boldsymbol{T}\|_2^{1/2} + 2 \| \boldsymbol{\Delta}\|_2\right)^4
        = \mathcal{K}_{\boldsymbol{\Delta}, \boldsymbol{T}}^4.
    \end{equation}
    Consequently,
    \begin{equation}
        \begin{aligned}
            &\mathbb{P}\left(\left|\frac{1}{n}\sum_{l=1}^{n} H^{(l)}\right| > t\right) \\
            \leq& C_H \exp\left(-c_H n \min \left\{\frac{t^2}{\mathcal{K}_{\boldsymbol{\Delta}, \boldsymbol{T}}^4 \|\boldsymbol{\Lambda}\|_F^2}, 1 \right\}\right),
        \end{aligned}
    \end{equation}
    where $C_H = 4$ and $c_H = \min\{c_{\mathrm{hw}}, c_{H,1}\}$.

    \paragraph{Control of $R^{(l)}$}
    Conditionally on $\mathcal{S}$, the entries of $\boldsymbol{\xi}_\Omega^{(l)}$ are independent and centered. 
    Furthermore, the triangular dithering property gives
    \begin{equation}
        \mathbb{E}\left[\boldsymbol{\xi}^{(l)} \boldsymbol{\xi}^{(l)^T} \middle| \mathcal{S}\right] = \dfrac{\boldsymbol{D}_{\boldsymbol{\Delta}}^2}{4}.
    \end{equation}
    Therefore,
    \begin{equation}
        \mathbb{E}\left[\boldsymbol{\xi}^{(l)^T} \boldsymbol{\Lambda} \boldsymbol{\xi}^{(l)} \middle| \mathcal{S}\right] = \dfrac{1}{4} \mathrm{tr}\left(\boldsymbol{\Lambda} \boldsymbol{D}_{\boldsymbol{\Delta}}^2\right),
    \end{equation}
    which is deterministic. 
    Hence
    \begin{equation}
        \mathbb{E}\left[\boldsymbol{\xi}^{(l)^T} \boldsymbol{\Lambda} \boldsymbol{\xi}^{(l)} \middle| \mathcal{S}\right] = \mathbb{E}\left[\boldsymbol{\xi}^{(l)^T} \boldsymbol{\Lambda} \boldsymbol{\xi}^{(l)}\right].
    \end{equation}

    Note that
    \begin{equation}
        R^{(l)} = \boldsymbol{\xi}_\Omega^{(l)^T} \boldsymbol{\Lambda} \boldsymbol{\xi}_\Omega^{(l)} - \mathbb{E}\left[\boldsymbol{\xi}_\Omega^{(l)^T} \boldsymbol{\Lambda} \boldsymbol{\xi}_\Omega^{(l)} \middle| \mathcal{S}\right].
    \end{equation}
    Conditionally on $\mathcal{S}$, the coordinates of $\boldsymbol{\xi}$ are independent, centered, and sub-Gaussian with
    \begin{equation}
        \begin{aligned}
            \max_j \|\xi^{(l)}_{\Omega, j} | \mathcal{S}\|_{\psi_2} 
            \leq& 2 \max_j \Delta_j 
            \leq 2 \max\{\Delta_r, \Delta_i\} \\
            &\leq 2\|\boldsymbol{\Delta}\|_2 \leq \mathcal{K}_{\boldsymbol{\Delta}, \boldsymbol{T}}.
        \end{aligned}
    \end{equation}
    Thus, by the Hanson--Wright inequality \cite{rudelson2013hanson} and the fact $\|\boldsymbol{\Lambda}\|_2 \leq \|\boldsymbol{\Lambda}\|_F$, we obtain
    \begin{equation}
        \begin{aligned}
            &\mathbb{P}\left(\left|\frac{1}{n} \sum_{l=1}^{n} R^{(l)}\right| > t \middle| \mathcal{S}\right) \\
            =& \mathbb{P}\left(\left|\frac{1}{n} \sum_{l=1}^{n} \boldsymbol{\xi}_\Omega^{(l)^T} \boldsymbol{\Lambda} \boldsymbol{\xi}_\Omega^{(l)} - \mathbb{E}\left[\boldsymbol{\xi}_\Omega^{(l)^T} \boldsymbol{\Lambda} \boldsymbol{\xi}_\Omega^{(l)}\middle| \mathcal{S}\right]\right| > t \middle| \mathcal{S}\right) \\
            \leq& 2 \exp \left(-c_{\mathrm{hw}} n \min\left\{\frac{t^2}{\mathcal{K}_{\boldsymbol{\Delta}, \boldsymbol{T}}^4\|\boldsymbol{\Lambda}\|_F^2}, \frac{t}{\mathcal{K}_{\boldsymbol{\Delta}, \boldsymbol{T}}^2\|\boldsymbol{\Lambda}\|_2}\right\}\right) \\
            \leq& 2 \exp \left(-c_{\mathrm{hw}} n \min\left\{\frac{t^2}{\mathcal{K}_{\boldsymbol{\Delta}, \boldsymbol{T}}^4 \|\boldsymbol{\Lambda}\|_F^2}, 1\right\}\right).
        \end{aligned}
    \end{equation}
    Since the right-hand side is independent of $\mathcal{S}$, the same bound holds unconditionally.

    Combining the three bounds, we have
    \begin{equation}
        \begin{aligned}
            &\mathbb{P}\left(\left|\frac{1}{n} \sum_{l=1}^{n} \dot{\boldsymbol{u}}_\Omega^{(l)^T} \boldsymbol{\Lambda} \dot{\boldsymbol{u}}_\Omega^{(l)} - \mathbb{E}\left[\dot{\boldsymbol{u}}_\Omega^{(l)^T} \boldsymbol{\Lambda} \dot{\boldsymbol{u}}_\Omega^{(l)}\right]\right| > t\right) \\
            \leq& \mathbb{P}\left(\left|\frac{1}{n} \sum_{l=1}^{n} G^{(l)}\right| > \frac{t}{3}\right)
            + \mathbb{P}\left(\left|\frac{1}{n} \sum_{l=1}^{n} H^{(l)}\right| > \frac{t}{3}\right) \\
            &+ \mathbb{P}\left(\left|\frac{1}{n} \sum_{l=1}^{n} R^{(l)}\right| > \frac{t}{3}\right) \\
            \leq& C_{\mathrm{qhw}} \exp \left(-c_{\mathrm{qhw}} n \min\left\{\frac{t^2}{\mathcal{K}_{\boldsymbol{\Delta}, \boldsymbol{T}}^4 \|\boldsymbol{\Lambda}\|_F^2}, 1\right\}\right),
        \end{aligned}
    \end{equation}
    where $C_{\mathrm{qhw}} = 8$ and $c_{\mathrm{qhw}} = \min\{c_{\mathrm{hw}}, c_H\} / 9$ are universal constants. 
    This completes the proof.
\end{proof}

Equipped with these lemmas, we now proceed to the proof, which is divided into the following three steps.

\textbf{Firstly}, we investigate the non-asymptotic properties of the coefficient estimators $\widehat{\gamma}_s$.
This has already been established in Theorem~\ref{Th_bound_gamma}.
In particular, by applying \eqref{Eq_Bound_For_Gamma} and taking a union bound over $s \in [d]$, we obtain
\begin{equation}
    \begin{aligned}
        \mathbb{P}&\bigl(\exists\, s \in [d] : \lvert \widehat{\gamma}_s - \gamma_s \rvert > t \bigr) \\
        &\leq 8 d \exp{\left( -c_{\mathrm{lag}} n \min \left\{\dfrac{t^2}{C_{\mathrm{lag}} \mathcal{K}_{\boldsymbol{\Delta}, \boldsymbol{T}}^4}, 1\right\} \right)}.
    \end{aligned}
\end{equation}

\textbf{Secondly}, we analyze the non-asymptotic behavior of $L_{\widehat{\boldsymbol{\gamma}} - \boldsymbol{\gamma}} (\theta)$.
For a fixed $\theta \in [0, 1]$, we associate $L_{\widehat{\boldsymbol{\gamma}} - \boldsymbol{\gamma}}(\theta)$ with a Hermitian matrix $\boldsymbol{M} \in \mathbb{C}^{d \times d}$ defined by
\begin{equation}
    M_{j, k} = \dfrac{e^{i 2 \pi s \theta}}{|\Omega_s|}, \quad 1 \leq j \leq k \leq d,
\end{equation}
where $s = k - j$ and $M_{k, j} = M_{j, k}^*$. It then follows that
\begin{equation}
    \begin{aligned}
        L_{\widehat{\boldsymbol{\gamma}} - \boldsymbol{\gamma}}(\theta) &= e_0 + \sum_{s = 1}^{d - 1} (\widehat{\gamma}_s - \gamma_s) e^{i 2 \pi s \theta} + \sum_{s = 1}^{d - 1} (\widehat{\gamma}_s - \gamma_s)^* e^{- i 2 \pi s \theta}\\
        &= \text{tr}(\widehat{\boldsymbol{T}}_\Omega - \boldsymbol{T}_\Omega, \boldsymbol{M}_\Omega).
    \end{aligned}
\end{equation}
Furthermore, let $\overline{\boldsymbol{T}} = \frac{1}{n} \sum_{l = 1}^n \dot{\boldsymbol{z}}^{(l)} \dot{\boldsymbol{z}}^{(l)^*} - \frac{\| \boldsymbol{\Delta}\|_2^2}{4} \boldsymbol{I}_d$.
We then observe that $\sum_{(j, k) \in \Omega_s} \widehat{T}_{j, k} = \sum_{(j, k) \in \Omega_s} \overline{T}_{j, k}$, and hence
\begin{equation}
    \begin{aligned}
        \text{tr}(\widehat{\boldsymbol{T}}_\Omega, \boldsymbol{M}_\Omega) &= \sum_{(j, k) \in \Omega^2} \widehat{T}_{j, k} M_{j, k} = \sum_{s = 1-d}^{d-1} \sum_{(j, k) \in \Omega_s} \widehat{T}_{j, k} M_{j, k}\\
        &= \sum_{s = 1-d}^{d-1} \sum_{(j, k) \in \Omega_s} \overline{T}_{j, k} M_{j, k} = \text{tr}(\overline{\boldsymbol{T}}_\Omega, \boldsymbol{M}_\Omega).
    \end{aligned}
\end{equation}
The third equality follows from the fact that $M_{j, k}$ depends only on the index difference $k-j$.
Let
\begin{equation}
    \dot{\boldsymbol{\Sigma}} = \mathbb{E}\left[\dot{\boldsymbol{z}}^{(l)} \dot{\boldsymbol{z}}^{(l)^H} \right] = \boldsymbol{T} + \dfrac{\|\boldsymbol{\Delta\|_2^2}}{4} \boldsymbol{I}_d,
\end{equation}
then we have $\mathbb{E} [ \dot{\boldsymbol{z}}_\Omega^{(l)^H} \boldsymbol{M}_\Omega \dot{\boldsymbol{z}}_\Omega^{(l)}] = \text{tr}(\dot{\boldsymbol{\Sigma}}_\Omega, \boldsymbol{M}_\Omega)$, which leads to
\begin{equation}
    \begin{aligned}
        L_{\widehat{\boldsymbol{\gamma}} - \boldsymbol{\gamma}}(\theta) &= \text{tr}(\widehat{\boldsymbol{T}}_\Omega - \boldsymbol{T}_\Omega, \boldsymbol{M}_\Omega)\\
        &= \dfrac{1}{n} \sum_{l = 1}^n \dot{\boldsymbol{z}}_\Omega^{(l)^H} \boldsymbol{M}_\Omega \dot{\boldsymbol{z}}_\Omega^{(l)} - \text{tr}(\dot{\boldsymbol{\Sigma}}_\Omega, \boldsymbol{M}_\Omega)\\
        &= \dfrac{1}{n} \sum_{l=1}^n \left(\dot{\boldsymbol{z}}_\Omega^{(l)^H} \boldsymbol{M}_\Omega \dot{\boldsymbol{z}}_\Omega^{(l)} - \mathbb{E} \left[ \dot{\boldsymbol{z}}_\Omega^{(l)^H} \boldsymbol{M}_\Omega \dot{\boldsymbol{z}}_\Omega^{(l)} \right]\right).
    \end{aligned}
\end{equation}
Since $\boldsymbol{M}_\Omega$ is Hermitian, we have
\begin{equation}
    \left(\dot{\boldsymbol{z}}_\Omega^{(l)^H} \boldsymbol{M}_\Omega \dot{\boldsymbol{z}}_\Omega^{(l)}\right)^*
    = \dot{\boldsymbol{z}}_\Omega^{(l)^H} \boldsymbol{M}_\Omega^H \dot{\boldsymbol{z}}_\Omega^{(l)}
    = \dot{\boldsymbol{z}}_\Omega^{(l)^H} \boldsymbol{M}_\Omega \dot{\boldsymbol{z}}_\Omega^{(l)}.
\end{equation}
Thus $\dot{\boldsymbol{z}}_\Omega^{(l)^H} \boldsymbol{M}_\Omega \dot{\boldsymbol{z}}_\Omega^{(l)}$ is real-valued.
If we write $\boldsymbol{M}_{\Omega} = \boldsymbol{M}_{r, \Omega} + i \boldsymbol{M}_{i, \Omega}$, then
\begin{equation}
    \begin{aligned}
        &\dot{\boldsymbol{z}}_\Omega^{(l)^H} \boldsymbol{M}_\Omega \dot{\boldsymbol{z}}_\Omega^{(l)}\\
        =& \left(\dot{\boldsymbol{x}}_\Omega^{(l)^T} - i \dot{\boldsymbol{y}}_\Omega^{(l)^T} \right)  \left(\boldsymbol{M}_{r, \Omega} + i \boldsymbol{M}_{i, \Omega}\right) \left(\dot{\boldsymbol{x}}_\Omega^{(l)} + i \dot{\boldsymbol{y}}_\Omega^{(l)}\right)\\
        =& \dot{\boldsymbol{u}}_\Omega^{(l)^T} \boldsymbol{\Lambda} \dot{\boldsymbol{u}}_\Omega^{(l)}, 
    \end{aligned}
\end{equation}
where $\dot{\boldsymbol{u}}_\Omega^{(l)} = (\dot{\boldsymbol{x}}_\Omega^{(l)^T}, \dot{\boldsymbol{y}}_\Omega^{(l)^T})^T$, and $\boldsymbol{\Lambda} = \begin{pmatrix} \boldsymbol{M}_{r, \Omega} & -\boldsymbol{M}_{i, \Omega}\\ \boldsymbol{M}_{i, \Omega} & \boldsymbol{M}_{r, \Omega}\end{pmatrix}$

Applying Lemma~\ref{Lem_W_H_type_inequality}, we obtain that, for any $t > 0$,
\begin{equation}
    \label{Eq_Bound_For_L}
    \begin{aligned}
        &\mathbb{P}\left( |L_{\widehat{\boldsymbol{\gamma}} - \boldsymbol{\gamma}}(\theta)| > t\right) \\
        =& \mathbb{P} \left( \left| \dfrac{1}{n} \sum_{l=1}^n \left(\dot{\boldsymbol{z}}_\Omega^{(l)^H} \boldsymbol{M}_\Omega \dot{\boldsymbol{z}}_\Omega^{(l)} - \mathbb{E} \left[ \dot{\boldsymbol{z}}_\Omega^{(l)^H} \boldsymbol{M}_\Omega \dot{\boldsymbol{z}}_\Omega^{(l)} \right]\right)\right| > t\right)\\
        \leq& C_{\mathrm{qhw}} \exp{\left( -c_{\mathrm{qhw}} n \min \left\{ \dfrac{t^2}{\mathcal{K}_{\boldsymbol{\Delta}, \boldsymbol{T}}^4 \|\boldsymbol{\Lambda}\|_F^2}, 1 \right\} \right)}\\
        \leq& C_{\mathrm{qhw}} \exp{\left( -c_{\mathrm{qhw}} n \min \left\{ \dfrac{t^2}{8 \mathcal{K}_{\boldsymbol{\Delta}, \boldsymbol{T}}^4 \phi(\Omega)}, 1 \right\} \right)},
    \end{aligned}
\end{equation}
where the last inequality follows from
\begin{equation}
    \begin{aligned}
        \|\boldsymbol{\Lambda}\|_F^2 &= 2\|\boldsymbol{M}_{r}\|_F^2 + 2 \|\boldsymbol{M}_{i}\|_F^2\\
        &\leq 4 \sum_{s=0}^{d-1} |\Omega_s| \cdot \dfrac{\cos^2(2\pi s\theta)}{|\Omega_s|^2} + 4 \sum_{s=0}^{d-1} |\Omega_s| \cdot \dfrac{\sin^2(2\pi s\theta)}{|\Omega_s|^2}\\
        &\leq 8 \sum_{s=0}^{d-1} \dfrac{1}{|\Omega_s|}
        = 8 \phi(\Omega).
    \end{aligned}
\end{equation}

\textbf{Finally}, we analyze the non-asymptotic performance of the Toeplitz covariance estimator $\widehat{\boldsymbol{T}}$.
We use a standard $\epsilon$-net argument to complete the proof.
Discretise the interval $[0,1]$ by the finite grid
\begin{equation}
    N = \left\{0, \dfrac{1}{4 \pi d^2}, \dfrac{2}{4 \pi d^2}, \ldots, 1\right\}.
\end{equation}
then, by the bound established in \eqref{Eq_Bound_For_L}, we obtain
\begin{equation}
    \begin{aligned}
        \mathbb{P}&\left(\exists \theta \in N: |L_{\widehat{\boldsymbol{\gamma}} - \boldsymbol{\gamma}}(\theta)| > \dfrac{t}{2}\right) \\
        &\leq C_{\mathrm{sdf}} d^2 \exp{\left( -c_{\mathrm{qhw}} n \min \left\{ \dfrac{t^2}{32 \mathcal{K}_{\boldsymbol{\Delta}, \boldsymbol{T}}^4 \phi(\Omega)}, 1 \right\} \right)},
    \end{aligned}
\end{equation}
where $C_{\mathrm{sdf}} > 4\pi C_{\mathrm{qhw}}$.
Define the event
\begin{equation}
\mathcal{A}:\left\{
\begin{aligned}
&\forall s \in [d]:\; |\widehat{\gamma}_s - \gamma_s| \leq t,\\
&\forall \theta \in N:\; |L_{\widehat{\boldsymbol{\gamma}} - \boldsymbol{\gamma}}(\theta)|
   \leq \frac{t}{2}.
\end{aligned}
\right.
\end{equation}
then it follows that, for sufficiently small $t > 0$,
\begin{equation}
    \begin{aligned}
        \mathbb{P}\left(\mathcal{A}\right)
        \geq& 1 - 8 d \exp{\left( -c_{\mathrm{lag}} n \min \left\{\dfrac{t^2}{C_{\mathrm{lag}} \mathcal{K}_{\boldsymbol{\Delta}, \boldsymbol{T}}^4}, 1\right\} \right)}\\
        &- C_{\mathrm{sdf}} d^2 \exp{\left( -c_{\mathrm{qhw}} n \min \left\{ \dfrac{t^2}{32 \mathcal{K}_{\boldsymbol{\Delta}, \boldsymbol{T}}^4 \phi(\Omega)}, 1 \right\} \right)}\\
        \geq& 1 - C_{\mathrm{op}} d^2 \exp{\left( -c_{\mathrm{op}} n \min \left\{\dfrac{t^2}{\mathcal{K}_{\boldsymbol{\Delta}, \boldsymbol{T}}^4 \phi(\Omega)}, 1\right\} \right)},\\
    \end{aligned}
\end{equation}
where $C_{\mathrm{op}} = 8 + C_{\mathrm{sdf}}$ and $c_{\mathrm{op}} = \min \left\{\frac{c_{\mathrm{lag}}}{C_{\mathrm{lag}}}, \frac{c_{\mathrm{qhw}}}{32}\right\}$.
For any $\theta \in [0, 1]$, there exists a grid point $\theta' \in N$ such that $\theta - \theta' \in \left[0, \frac{1}{4 \pi d^2}\right)$.
On the event $\mathcal{A}$, we have
\begin{equation}
    \begin{aligned}
        | L_{\widehat{\boldsymbol{\gamma}} - \boldsymbol{\gamma}}(\theta) | &\leq | L_{\widehat{\boldsymbol{\gamma}} - \boldsymbol{\gamma}}(\theta') | + | \theta - \theta'| \sup_{\alpha \in [\theta', \theta]} | L'_{\widehat{\boldsymbol{\gamma}} - \boldsymbol{\gamma}}(\alpha) |\\
        &\leq \dfrac{t}{2} + \dfrac{1}{4 \pi d^2} 2\pi d^2 t = t,
    \end{aligned}
\end{equation}
where the last inequality uses the fact that
\begin{equation}
    \begin{aligned}
        | L'_{\widehat{\boldsymbol{\gamma}} - \boldsymbol{\gamma}}(\alpha) | &= \left| 2 \pi i \sum_{s=1-d}^{d-1} s \left(\widehat{\gamma}_s - \gamma_s\right) e^{i 2\pi s\alpha} \right| \\
        &\leq 2 \pi d^2 \max_{s\in [d]}| \widehat{\gamma}_s - \gamma_s | \leq 2 \pi d^2 t.
    \end{aligned}
\end{equation}
In conclusion, we arrive at the bound
\begin{equation}
    \begin{aligned}
        &\mathbb{P}\left(\exists \theta \in [0, 1]: |L_{\widehat{\boldsymbol{\gamma}} - \boldsymbol{\gamma}}(\theta)| > t \right) \\
        &\leq C_{\mathrm{op}} d^2 \exp{\left( -c_{\mathrm{op}} n \min \left\{\dfrac{t^2}{\mathcal{K}_{\boldsymbol{\Delta}, \boldsymbol{T}}^4 \phi(\Omega)}, 1\right\} \right)},
    \end{aligned}
\end{equation}
which immediately implies \eqref{Eq_Bound_For_T} by using Lemma~\ref{La_spectral_density} and taking $c = c_{\mathrm{op}}$, $C = C_{\mathrm{op}}$.

\subsection{Proof of Theorem~\ref{Th_bound_2k_bit_T}}\label{Pr_bound_2k_bit_T}
The analysis here is analogous to that in \cite[Appendix~F]{xu2024bit}.
We therefore focus only on the proof of \eqref{Eq_bound_2k_bit_T}; once \eqref{Eq_bound_2k_bit_T} is established, the remaining statements follow directly.

Note that both $\Re(\tau_j^{(l)})$ and $\Im(\tau_j^{(l)})$ are supported on $[-\Delta, \Delta]$, and hence
\begin{equation}
    |\tau_j^{(l)}| = \sqrt{\left(\Re(\tau_j^{(l)})\right)^2 + \left(\Im(\tau_j^{(l)})\right)^2} \leq \sqrt{2}\Delta. 
\end{equation}
Therefore, \eqref{Eq_max_z_plus_tau} is guaranteed if
\begin{equation}
    \max_{j \in \Omega, 1 \leq l \leq n} \left\{ \left| z_j^{(l)} \right| \right\} < (2^{k-1} - \sqrt{2}) \Delta.
\end{equation}
Since $z_j^{(l)} \sim \mathcal{CN}(0,\gamma_0)$, the magnitude $\bigl|z_j^{(l)}\bigr|$ follows a Rayleigh distribution.
Thus, for any $t > 0$,
\begin{equation}
    \mathbb{P}\left(\left| z_j^{(l)} \right| > t\right) = \exp\left(-\dfrac{t^2}{\gamma_0}\right).
\end{equation}
Taking a union bound over $j\in \Omega$ and $1 \leq l \leq n$ yields
\begin{equation}
    \label{Eq_T_T_k}
    \mathbb{P}\left(\max_{j \in \Omega, 1 \leq l \leq n} \left\{ \left| z_j^{(l)} \right| \right\} > t\right) \leq n |\Omega| \exp\left(-\dfrac{t^2}{\gamma_0}\right).
\end{equation}
Let
\begin{equation}
    \Delta = C_{\mathrm{bit}} \cdot 2^{2-k} \sqrt{\gamma_0 \left(\log\left(n |\Omega|\right) + \delta'\right)}
\end{equation}
with $C_{\mathrm{bit}} \geq \frac{1}{2 - \sqrt{2}}$, and set $t = (2^{k-1} - \sqrt{2}) \Delta$ in \eqref{Eq_T_T_k}.
We then obtain
\begin{equation}
    \label{Eq_P_T_neq_T_k}
    \begin{aligned}
        \mathbb{P}\left(\widehat{\boldsymbol{T}}_{2k} \neq \widehat{\boldsymbol{T}}\right) &\leq \mathbb{P}\left(\max_{j \in \Omega, 1 \leq l \leq n} \left\{ \left| z_j^{(l)} \right| \right\} > (2^{k-1} - \sqrt{2}) \Delta\right)\\
        &\leq e^{-\delta'}.
    \end{aligned}
\end{equation}
In other words, with probability at least $1 - e^{-\delta'}$, the estimators
$\widehat{T}_{2k}$ and $\widehat{T}$ coincide and hence share the same error bound.
Combining \eqref{Eq_P_T_neq_T_k} with the bound in \eqref{Eq_error_T} and taking $C = C_{\mathrm{bit}}$, we
immediately obtain that
\begin{equation}
    \|\widehat{\boldsymbol{T}}_{2k} - \boldsymbol{T}\|_2 \lesssim \mathcal{K}_{\boldsymbol{\Delta}, \boldsymbol{T}}^2 \sqrt{\dfrac{\phi(\Omega) \cdot \delta \log d}{n}}
\end{equation}
holds with probability at least $1-d^{2-\delta}-e^{-\delta'}$.

\bibliographystyle{IEEEtran}
\bibliography{References}

@STRING{icassp    = {Proc. IEEE Intern.  Conf. on Acoust., Speech and Signal Process. }}

@article{xu2024bit,
  title={Bit-efficient {Toeplitz} covariance estimation},
  author={Xu, Hongwei and Yang, Zai},
  journal={IEEE Transactions on Information Theory, DOI: 10.1109/TIT.2026.3697612},
  year={2026},
  publisher={IEEE}
}

@article{chen2023quantizing,
  title={Quantizing heavy-tailed data in statistical estimation:(near) minimax rates, covariate quantization, and uniform recovery},
  author={Chen, Junren and Ng, Michael K and Wang, Di},
  journal={IEEE Transactions on Information Theory},
  volume={70},
  number={3},
  pages={2003--2038},
  year={2023},
  publisher={IEEE}
}

@article{chen2025parameter,
  title={A parameter-free two-bit covariance estimator with improved operator norm error rate},
  author={Chen, Junren and Ng, Michael K},
  journal={Applied and Computational Harmonic Analysis},
  pages={101774},
  year={2025},
  publisher={Elsevier}
}

@article{chen2023quantized,
  title={Quantized low-rank multivariate regression with random dithering},
  author={Chen, Junren and Wang, Yueqi and Ng, Michael K},
  journal={IEEE Transactions on Signal Processing},
  volume={71},
  pages={3913--3928},
  year={2023},
  publisher={IEEE}
}

@article{dirksen2025subspace,
  title={Subspace and DOA estimation under coarse quantization},
  author={Dirksen, Sjoerd and Li, Weilin and Maly, Johannes},
  journal={IEEE Transactions on Information Theory},
  year={2025},
  publisher={IEEE}
}

@article{gray2002dithered,
  title={Dithered quantizers},
  author={Gray, Robert M and Stockham, Thomas G},
  journal={IEEE Transactions on Information Theory},
  volume={39},
  number={3},
  pages={805--812},
  year={1993},
  publisher={IEEE}
}

@article{aubry2021structured,
  title={Structured covariance matrix estimation with missing-(complex) data for radar applications via expectation-maximization},
  author={Aubry, Augusto and De Maio, Antonio and Marano, Stefano and Rosamilia, Massimo},
  journal={IEEE Transactions on Signal Processing},
  volume={69},
  pages={5920--5934},
  year={2021},
  publisher={IEEE}
}

@inproceedings{eldar2020sample,
  title={Sample efficient toeplitz covariance estimation},
  author={Eldar, Yonina C and Li, Jerry and Musco, Cameron and Musco, Christopher},
  booktitle={Proceedings of the Fourteenth Annual ACM-SIAM Symposium on Discrete Algorithms},
  pages={378--397},
  year={2020},
  organization={SIAM}
}

@inproceedings{wu2016direction,
  title={Direction-of-arrival estimation based on Toeplitz covariance matrix reconstruction},
  author={Wu, Xiaohuan and Zhu, Wei-Ping and Yan, Jun},
  booktitle={2016 IEEE International Conference on Acoustics, Speech and Signal Processing (ICASSP)},
  pages={3071--3075},
  year={2016},
  organization={IEEE}
}

@article{yang2014discretization,
  title={A discretization-free sparse and parametric approach for linear array signal processing},
  author={Yang, Zai and Xie, Lihua and Zhang, Cishen},
  journal={IEEE Transactions on Signal Processing},
  volume={62},
  number={19},
  pages={4959--4973},
  year={2014},
  publisher={IEEE}
}

@article{qiao2017gridless,
  title={Gridless line spectrum estimation and low-rank Toeplitz matrix compression using structured samplers: A regularization-free approach},
  author={Qiao, Heng and Pal, Piya},
  journal={IEEE Transactions on Signal Processing},
  volume={65},
  number={9},
  pages={2221--2236},
  year={2017},
  publisher={IEEE}
}

@article{yang2022nonasymptotic,
  title={Nonasymptotic performance analysis of ESPRIT and spatial-smoothing ESPRIT},
  author={Yang, Zai},
  journal={IEEE Transactions on Information Theory},
  volume={69},
  number={1},
  pages={666--681},
  year={2022},
  publisher={IEEE}
}

@book{vershynin2018high,
  title={High-dimensional probability: An introduction with applications in data science},
  author={Vershynin, Roman},
  volume={47},
  year={2018},
  publisher={Cambridge university press}
}

@book{grenander1958toeplitz,
  title={Toeplitz forms and their applications},
  author={Grenander, Ulf},
  year={1958},
  publisher={Univ of California Press}
}

@article{rudelson2013hanson,
  author = {Mark Rudelson and Roman Vershynin},
  title = {{Hanson-Wright inequality and sub-gaussian concentration}},
  volume = {18},
  journal = {Electronic Communications in Probability},
  number = {none},
  publisher = {Institute of Mathematical Statistics and Bernoulli Society},
  pages = {1 -- 9},
  year = {2013}
}

@article{dirksen2024tuning,
  title={Tuning-free one-bit covariance estimation using data-driven dithering},
  author={Dirksen, Sjoerd and Maly, Johannes},
  journal={IEEE Transactions on Information Theory},
  volume={70},
  number={7},
  pages={5228--5247},
  year={2024},
  publisher={IEEE}
}

@article{stoica2010spice,
  title={SPICE: A sparse covariance-based estimation method for array processing},
  author={Stoica, Petre and Babu, Prabhu and Li, Jian},
  journal={IEEE Transactions on Signal Processing},
  volume={59},
  number={2},
  pages={629--638},
  year={2010},
  publisher={IEEE}
}

@article{krim2002two,
  title={Two decades of array signal processing research: the parametric approach},
  author={Krim, Hamid and Viberg, Mats},
  journal={IEEE signal processing magazine},
  volume={13},
  number={4},
  pages={67--94},
  year={2002},
  publisher={IEEE}
}

@article{yang2019source,
  title={Source resolvability of spatial-smoothing-based subspace methods: A hadamard product perspective},
  author={Yang, Zai and Stoica, Petre and Tang, Jinhui},
  journal={IEEE Transactions on Signal Processing},
  volume={67},
  number={10},
  pages={2543--2553},
  year={2019},
  publisher={IEEE}
}

@article{johnstone2001distribution,
  title={On the distribution of the largest eigenvalue in principal components analysis},
  author={Johnstone, Iain M},
  journal={The Annals of statistics},
  volume={29},
  number={2},
  pages={295--327},
  year={2001},
  publisher={Institute of Mathematical Statistics}
}

@book{wright2022high,
  title={High-dimensional data analysis with low-dimensional models: Principles, computation, and applications},
  author={Wright, John and Ma, Yi},
  year={2022},
  publisher={Cambridge University Press}
}

@incollection{yang2018sparse,
  title={Sparse methods for direction-of-arrival estimation},
  author={Yang, Zai and Li, Jian and Stoica, Petre and Xie, Lihua},
  booktitle={Academic Press Library in Signal Processing, Volume 7},
  pages={509--581},
  year={2018},
  publisher={Elsevier}
}

@article{cai2013optimal,
  title={Optimal rates of convergence for estimating Toeplitz covariance matrices},
  author={Cai, T Tony and Ren, Zhao and Zhou, Harrison H},
  journal={Probability Theory and Related Fields},
  volume={156},
  number={1},
  pages={101--143},
  year={2013},
  publisher={Springer}
}

@article{gonen2016subspace,
  title={Subspace learning with partial information},
  author={Gonen, Alon and Rosenbaum, Dan and Eldar, Yonina C and Shalev-Shwartz, Shai},
  journal={Journal of Machine Learning Research},
  volume={17},
  number={52},
  pages={1--21},
  year={2016}
}

@article{yang2023robust,
  title={A robust and statistically efficient maximum-likelihood method for DOA estimation using sparse linear arrays},
  author={Yang, Zai and Chen, Xinyao and Wu, Xunmeng},
  journal={IEEE Transactions on Aerospace and Electronic Systems},
  volume={59},
  number={5},
  pages={6798--6812},
  year={2023},
  publisher={IEEE}
}

@article{du2020toeplitz,
  title={Toeplitz structured covariance matrix estimation for radar applications},
  author={Du, Xiaolin and Aubry, Augusto and De Maio, Antonio and Cui, Guolong},
  journal={IEEE Signal Processing Letters},
  volume={27},
  pages={595--599},
  year={2020},
  publisher={IEEE}
}

@article{aubry2024advanced,
  title={Advanced methods for MLE of Toeplitz structured covariance matrices with applications to radar problems},
  author={Aubry, Augusto and Babu, Prabhu and De Maio, Antonio and Rosamilia, Massimo},
  journal={IEEE Transactions on Information Theory},
  year={2024},
  publisher={IEEE}
}

@article{patole2017automotive,
  title={Automotive radars: A review of signal processing techniques},
  author={Patole, Sujeet Milind and Torlak, Murat and Wang, Dan and Ali, Murtaza},
  journal={IEEE Signal Processing Magazine},
  volume={34},
  number={2},
  pages={22--35},
  year={2017},
  publisher={IEEE}
}

@article{liao2012doa,
  title={DOA estimation and tracking of ULAs with mutual coupling},
  author={Liao, Bin and Zhang, Zhi-Guo and Chan, Shing-Chow},
  journal={IEEE Transactions on Aerospace and Electronic Systems},
  volume={48},
  number={1},
  pages={891--905},
  year={2012},
  publisher={IEEE}
}

@article{li2025sparse,
  title={Sparse linear arrays for direction-of-arrival estimation: A tutorial overview},
  author={Li, Xiang and Jin, Ming and Meng, Xiang-Tian and Cao, Bing-Xia and Yan, Feng-Gang and Greco, Maria Sabrina and Gini, Fulvio},
  journal={IEEE Aerospace and Electronic Systems Magazine},
  year={2025},
  publisher={IEEE}
}

@article{walden2002analog,
  title={Analog-to-digital converter survey and analysis},
  author={Walden, Robert H},
  journal={IEEE Journal on selected areas in communications},
  volume={17},
  number={4},
  pages={539--550},
  year={2002},
  publisher={IEEE}
}

@article{hu2025model,
  title={Model-Based Learning for DOA Estimation with One-Bit Single-Snapshot Sparse Arrays},
  author={Hu, Yunqiao and Sun, Shunqiao and Zhang, Yimin D},
  journal={arXiv preprint arXiv:2502.17473},
  year={2025}
}

@article{studer2016quantized,
  title={Quantized massive mu-mimo-ofdm uplink},
  author={Studer, Christoph and Durisi, Giuseppe},
  journal={IEEE Transactions on Communications},
  volume={64},
  number={6},
  pages={2387--2399},
  year={2016},
  publisher={IEEE}
}

@article{dirksen2022covariance,
  title={Covariance estimation under one-bit quantization},
  author={Dirksen, Sjoerd and Maly, Johannes and Rauhut, Holger},
  journal={The Annals of Statistics},
  volume={50},
  number={6},
  pages={3538--3562},
  year={2022},
  publisher={Institute of Mathematical Statistics}
}

@incollection{maly2022new,
  title={New challenges in covariance estimation: multiple structures and coarse quantization},
  author={Maly, Johannes and Yang, Tianyu and Dirksen, Sjoerd and Rauhut, Holger and Caire, Giuseppe},
  booktitle={Compressed Sensing in Information Processing},
  pages={77--104},
  year={2022},
  publisher={Springer}
}

@article{chen2023high,
  title={High dimensional statistical estimation under uniformly dithered one-bit quantization},
  author={Chen, Junren and Wang, Cheng-Long and Ng, Michael K and Wang, Di},
  journal={IEEE Transactions on Information Theory},
  volume={69},
  number={8},
  pages={5151--5187},
  year={2023},
  publisher={IEEE}
}

@article{lu20251,
  title={A 1.5-bit Quantization Scheme and Its Application to Direction Estimation},
  author={Lu, Xicheng and Liu, Wei and Alomainy, Akram},
  journal={IEEE Transactions on Signal Processing},
  year={2025},
  publisher={IEEE}
}

@article{yang2016vandermonde,
  title={Vandermonde decomposition of multilevel Toeplitz matrices with application to multidimensional super-resolution},
  author={Yang, Zai and Xie, Lihua and Stoica, Petre},
  journal={IEEE Transactions on Information Theory},
  volume={62},
  number={6},
  pages={3685--3701},
  year={2016},
  publisher={IEEE}
}

@article{mahot2013asymptotic,
  title={Asymptotic properties of robust complex covariance matrix estimates},
  author={Mahot, M{\'e}lanie and Pascal, Fr{\'e}d{\'e}ric and Forster, Philippe and Ovarlez, Jean-Philippe},
  journal={IEEE Transactions on Signal Processing},
  volume={61},
  number={13},
  pages={3348--3356},
  year={2013},
  publisher={IEEE}
}

@article{romberg2009compressive,
  title={Compressive sensing by random convolution},
  author={Romberg, Justin},
  journal={SIAM Journal on Imaging Sciences},
  volume={2},
  number={4},
  pages={1098--1128},
  year={2009},
  publisher={SIAM}
}

@article{romero2015compressive,
  title={Compressive covariance sensing: Structure-based compressive sensing beyond sparsity},
  author={Romero, Daniel and Ariananda, Dyonisius Dony and Tian, Zhi and Leus, Geert},
  journal={IEEE signal processing magazine},
  volume={33},
  number={1},
  pages={78--93},
  year={2015},
  publisher={IEEE}
}

@article{vandenberghe1996semidefinite,
  title={Semidefinite programming},
  author={Vandenberghe, Lieven and Boyd, Stephen},
  journal={SIAM review},
  volume={38},
  number={1},
  pages={49--95},
  year={1996},
  publisher={SIAM}
}

@article{gilbert1993increased,
  title={Increased information rate by oversampling},
  author={Gilbert, Edgar N.},
  journal={IEEE transactions on information theory},
  volume={39},
  number={6},
  pages={1973--1976},
  year={1993}
}

@inproceedings{koch2010increased,
  title={Increased capacity per unit-cost by oversampling},
  author={Koch, Tobias and Lapidoth, Amos},
  booktitle={2010 IEEE 26-th Convention of Electrical and Electronics Engineers in Israel},
  pages={000684--000688},
  year={2010},
  organization={IEEE}
}

@article{uccuncu2018oversampling,
  title={Oversampling in one-bit quantized massive MIMO systems and performance analysis},
  author={{\"U}{\c{c}}{\"u}nc{\"u}, Ali Bulut and Y{\i}lmaz, Ali {\"O}zg{\"u}r},
  journal={IEEE Transactions on Wireless Communications},
  volume={17},
  number={12},
  pages={7952--7964},
  year={2018},
  publisher={IEEE}
}

@article{shao2021dynamic,
  title={Dynamic oversampling for 1-bit ADCs in large-scale multiple-antenna systems},
  author={Shao, Zhichao and Landau, Lukas TN and de Lamare, Rodrigo C},
  journal={IEEE transactions on communications},
  volume={69},
  number={5},
  pages={3423--3435},
  year={2021},
  publisher={IEEE}
}

@article{shao2020channel,
  title={Channel estimation for large-scale multiple-antenna systems using 1-bit ADCs and oversampling},
  author={Shao, Zhichao and Landau, Lukas TN and De Lamare, Rodrigo C},
  journal={IEEE access},
  volume={8},
  pages={85243--85256},
  year={2020},
  publisher={IEEE}
}

\end{document}